\documentclass[sigplan]{acmart}
\usepackage{algorithmic}
\usepackage[ruled,vlined]{algorithm2e}
\usepackage{multirow}

\AtBeginDocument{%
  \providecommand\BibTeX{{%
    \normalfont B\kern-0.5em{\scshape i\kern-0.25em b}\kern-0.8em\TeX}}}

\setcopyright{acmcopyright}
\copyrightyear{2023}
\acmYear{2023}
\acmDOI{XXXXXXX.XXXXXXX}

\acmConference[XXXXXX 'XXX]{XXXXXX}{XXXXX}{XXXXX}
\begin{document}

\newcommand{\Ropt}{$\mathbf{r}_{\text{opt}}$~}

\title{Training DNN Models over Heterogeneous Clusters with Optimal Performance}

\author{Chengyi Nie}
\affiliation{%
  \institution{Stony Brook University}
  \city{Stony Brook}
  \state{NY}
  \country{USA}
}

\author{Jessica Maghakian}
\affiliation{%
  \institution{Stony Brook University}
  \city{Stony Brook}
  \state{NY}
  \country{USA}
}

\author{Zhenhua Liu}
\affiliation{%
  \institution{Stony Brook University}
  \city{Stony Brook}
  \state{NY}
  \country{USA}
}






\renewcommand{\shortauthors}{Nie, et al.}


\begin{abstract}
    Adjusting batch sizes and adaptively tuning other hyperparameters can significantly speed up deep neural network (DNN) training. Despite the ubiquity of heterogeneous clusters, existing adaptive DNN training techniques solely consider homogeneous environments.
    Optimizing distributed DNN training over heterogeneous clusters is technically challenging, and directly adapting existing techniques results in low utilization and poor performance. To solve this problem, we introduce Cannikin -- a novel data-parallel distributed training system. Cannikin achieves efficient and near optimal performance by accurately modeling the optimal system performance and predicting adaptive batch size training metrics for DNNs in heterogeneous clusters. 
    We implemented Cannikin in PyTorch and conducted experiments over 16 GPUs in Chameleon. Empirical results show that Cannikin reduces DNN training in heterogeneous clusters by up to $52\%$ compared to the state-of-art adaptive training system and up to $85\%$ compared to native PyTorch DistributedDataParallel. 
\end{abstract}





\maketitle

\section{Introduction}

With the explosive increase of deep learning (DL) applications in fields such as image classification~\cite{imagenet, cifar}, natural language processing~\cite{squad,googletranslate}, and recommender systems~\cite{movielens,recommendersystem}, the demand for deep neural network (DNN) training resources is doubling every six months~\cite{dltrend}. In order to achieve efficient DNN training, practitioners rely on accelerators~\cite{tpu, fpga,asic}, hyper-parameter tuning~\cite{gns}, and distributed training~\cite{li2020pytorch, tf}. Meanwhile, the short hardware update cycle results in newly released accelerators that significantly outperform previous models within a short time~\cite{moore}. Table~\ref{tab:gpus} shows the evolution of NVIDIA data center GPUs released in recent years. Each new flagship model is over two times faster than the preceding flagship data center GPU. When companies and research institutes upgrade their machine learning systems, newly released GPUs are installed before older models retire, so homogeneous environments cannot always be guaranteed when running distributed training jobs. To enhance the utilization of computing resources and speed up DNN model training in heterogeneous environments, specialized methods are required for DNN training in heterogeneous environments.

\begin{table}[h]
\caption{Evolution of NVIDIA data center GPUs}
  \small
  \begin{tabular}{cccccc}
    \toprule
    \multirow{2}* {\textbf{Model}} & \multirow{2}* {\textbf{Year}} &\multirow{2}* {\textbf{Archit.} } &\textbf{CUDA} & \textbf{Memory} & \textbf{FP16}  \\ 
    & & & \textbf{Cores}& (GB) & (TFLOPS) \\  
    \midrule
    Tesla P100 & 2016 & Pascal &3584 &16 &21.2 \\
    Tesla V100 & 2017 & Volta &5120 &16/32&31.4 \\
    A100 & 2020 & Ampere &6912 &40/80 &77.97 \\
    H100 & 2022 & Hopper &16896 &80&204.9 \\
  \bottomrule
  \end{tabular}
  \label{tab:gpus}
\end{table}


Previous work on specialized distributed training for heterogeneous environments focuses on two major schemes: data parallelism~\cite{li2020pytorch, horovod} and model parallelism~\cite{model_parall, hetpipe}. For data-parallelism heterogeneous distributed training systems, HetSeq~\cite{hetseq} manually tunes the local mini batch size for each node to balance workloads, while LB-BSP~\cite{socc20} and DLB~\cite{DLB} improve performance by iteratively tuning the workloads assigned to each worker based on the computing time of each node. On the other hand, BlueConnect~\cite{MLSYS2019_9b861925} boosts data-parallelism distributed training by optimizing communication. Existing data-parallelism systems do not jointly consider the computing and communication models for heterogeneous clusters, resulting in suboptimal performance. Model-parallelism systems~\cite{hetpipe, pipedream} pipeline the DNN model in heterogeneous clusters, which is near optimal for resource utilization with fine-tuning. But model parallelism requires a specific configuration of each node in a cluster, thus limiting scalability. Furthermore, existing model-parallelism and data-parallelism methods cannot manage the sudden changes of resources that occur in clusters with dynamic resource allocation~\cite{optimus, pollux}.

Another efficient DNN training method, adaptive batch size training, tunes hyper-parameters such as batch size and learning rate during training, significantly speeding up convergence. Previous work~\cite{pollux,adabatch,gns,kungfu} focuses on adaptive batch size training in homogeneous environments. These approaches continuously monitor system metrics and optimize the batch size according to their adaption policies. However, the adaptive batch size metrics, algorithms, and adaption policies are all designed for homogeneous environments. Directly adopting existing methods in heterogeneous clusters will cause a large margin of error for adaptive batch size training metrics measurement and prediction. To the best of our knowledge, no specialized adaptive training system for heterogeneous environments has been developed so far.

There are three main challenges in designing an automatic near-optimal training system. First, each node has a different performance model in heterogeneous clusters, which causes complexity in predicting the optimal performance and the corresponding cluster configuration. The system overhead will significantly affect the training performance for larger clusters. Second, considering data parallelism distributed training in heterogeneous clusters, given a new total batch size to the cluster in the adaptive batch size training, the optimal configuration of each node will change. Given the total batch size, previous work~\cite{socc20,DLB} iteratively tunes each node's configuration to approach the optimal performance, which is inefficient in adaptive batch size training. Third, in heterogeneous data parallelism training systems, different local batch sizes are assigned to different GPUs. This introduces challenges in accurately modeling gradient noise~\cite{gns} across the heterogeneous clusters.

Based on these insights and challenges, we study the performance model of data-parallel distributed training for heterogeneous GPU clusters and propose the optimal performance \textit{OptPerf}. For the prediction of \textit{OptPerf} and the corresponding configuration, we design Cannikin, an efficient near-optimal data-parallel distributed training system. Taking both computing and communication models into consideration, Cannikin has accurate modeling and prediction of the performance models for DL in heterogeneous clusters. With the cluster performance model learned online, Cannikin can predict \textit{OptPerf} and configuration with low overhead when the cluster is given a new total batch size. Cannikin optimizes the measurement of system parameters using inverse variance weighting of different observations from each node in the cluster. We also prove that in heterogeneous clusters we can correctly model adaptive batch size training metrics, just like systems~\cite{pollux, gns} designed for homogeneous clusters. While developed for single DNN training, Cannikin can be readily integrated with adaptive batch size training engines~\cite{pollux,kungfu} and dynamic resource allocation schedulers~\cite{pollux,optimus} for multiple jobs. The main contributions of this paper are:
\begin{itemize}
\item We are the first to consider training DNN models using adaptive batch sizes in heterogeneous clusters.
\item We deduce and predict the optimal performance, denoted as \textit{OptPerf}, along with its corresponding configuration and the optimal total batch sizes during DL training in heterogeneous clusters.
\item We design and implement Cannikin that can be easily integrated with the state-of-art adaptive batch size training systems to train the DNN models in heterogeneous clusters optimally.
\item We evaluate the performance of Cannikin in two real heterogeneous clusters using multiple popular DNN workloads and optimizers. Results highlight that Cannikin reduces DNN training by up to $52\%$ and $85\%$ in heterogeneous clusters compared to AdaptDL~\cite{pollux} and PyTorch DistributedDataParallel(DDP)~\cite{li2020pytorch}. 
\end{itemize}

\section{Background}
\subsection{Data-parallel Distributed Training}
DNN workloads are highly computing-intensive~\cite{imagenettrain}. Distributed deep learning accelerates DNN model training with multiple GPUs, by either using model parallelism and data parallelism. This paper focuses on data-parallelism distributed learning, which trains the same model on multiple GPUs and each GPU uses different data samples.  

During data-parallelism distributed training, node $i$ first trains its local mini batch by forward and backward passes for the local gradient estimation. Node $i$'s local gradient $g_i$ is:
\begin{equation}
    g_i = \frac{1}{b_i}\sum_{j=0}^{b_i-1}\nabla_{\theta} L_{x_j}(\theta),
\label{eq:localgradient}
\end{equation}
where $\theta$ is the vector of DNN weight parameters, $b_i$ is node~$i$'s local mini-batch size, $x_j$ is a sample in local mini batch, and $L$ is the loss function.

Upon the completion of the local gradient estimation, node~$i$ will aggregate its own gradient with all the local gradients calculated by other nodes in the cluster to get the gradient of the total batch (full batch of the DNN model):
\begin{equation}
    g = \frac{1}{N}\sum_{i=0}^{N-1}g_i,
\end{equation}
where $\textit{N}$ is the count of total nodes (GPUs) that train the model. Finally, each node will use the averaged gradient to update the weight parameters $\bf {w}$ for the next batch. In a DNN training system, the gradient averaging can be handled using backends such as NCCL~\cite{nccl}, MPI~\cite{mpi}, and Gloo~\cite{gloo}. 

\subsection{Adaptive Batch Size Training}
A deep learning model usually consists of millions to trillions of parameters~\cite{resnet,attentionisall,pathway}, requiring a long time for training. The selection of hyperparameters such as batch size in different convergence phases significantly impacts training efficiency. Recent work on adaptive batch size training~\cite{pollux,adabatch,gns,kungfu} greatly speeds up deep learning model training by dynamically tuning hyperparameters such as batch size and learning rate according to the gradient noise~\cite{gns}, data throughput, and other customized metrics.

To determine the most statistically efficient batch size for a training iteration, \citet{gns} propose the \emph{gradient noise scale}, an estimation of the signal-to-noise ratio of the stochastic gradient. This metric can be used to predict the most statistically efficient batch size during training. When the gradient noise is low, a small batch size can achieve a great contribution to the convergence. When the gradient noise is large, the error of the gradient estimated by a small batch would be significant. In this situation, a larger batch size could reduce the training time with little reduction of statistical efficiency.

However the optimal convergence progress is not guaranteed by using the most statistically efficient batch size, because the most statistically efficient batch size often comes with relatively low data throughput. Pollux~\cite{pollux} introduces \emph{goodput}, the product of the system throughput and statistical efficiency modeled by the gradient noise scale. Goodput optimizes training by balancing the trade-off between data throughput and statistical efficiency.

\section{The \emph{O\MakeLowercase{pt}P\MakeLowercase{erf}} of GPU Clusters}
To improve the performance of a cluster, first we need to know the optimal performance the cluster can achieve. In this section, we define and deduce the optimal performance \textit{OptPerf} of a heterogeneous GPU cluster using the metrics collected from each GPU in a general case.

\subsection{The Definition of \emph{OptPerf}}

For data-parallelism distributed training with a given total batch size $B$, \textit{OptPerf} is the optimal batch processing time a heterogeneous GPU cluster can achieve by ideally tuning each node's local mini-batch size. Consider heterogeneous GPU Cluster $A$ with a set of nodes $\mathcal{N}$, $|\mathcal{N}| = n$. For the standalone training of local mini batch size $b_i$ at node $i$, let the batch processing time without gradient synchronization among nodes be denoted by $t_i^{b_i}$. Due to hetergoeneity, we assume $t_i^{b}\neq t_j^{b}$ for any $b$ and $i\neq j$. The local mini batch sizes satisfy $\sum_{i\in\mathcal{N}}b_i = B$.
Define local mini batch size ratio ${\bf r}=[r_0, r_1, \dots, r_{n-1}]$, where $r_i = b_i / B$. In this paper, we only consider synchronized data-parallelism distributed training, meaning all nodes synchronize their updates after each batch's training. For this training model, fast nodes in the cluster always wait for the stragglers to finish local gradient estimation, hence the batch processing time $T = max\{t_0^{b_0}, t_1^{b_1}, \dots, t_{n-1}^{b_{n-1}}\}$. We can infer that there exists an optimal local mini batch size ratio \Ropt that minimizes the batch processing time $T$. We define the minimized batch processing time to be \textit{OptPerf}. To determine \textit{OptPerf} and \Ropt for heterogeneous cluster $A$ with total batch size $B$, we model the performance of GPUs in the cluster as a function of $T$, $B$, and ${\bf r}$.

\subsection{Performance Modeling}
When considering the performance model of a heterogeneous cluster, rather than simultaneously modeling the performance of heterogeneous nodes, we can instead model the performance of each GPU separately and then combine all the GPU performance models to determine the cluster's performance. In data parallel distributed training, the batch processing time is composed of the computing time for local gradient estimation and the communication time for gradient synchronization across nodes.

\subsubsection{Computing Time} The batch computing time can be separated into data loading, forward propagation, backward propagation and parameter updating. 
For any node $i$, the computing time $t_{compute}^i$ is a linear function of local batch size $b_i$~\cite{MLSYS2022_f457c545,pollux}. Furthermore, the parameter updating time remains constant regardless of variations in the local batch size. In contrast, data loading time, forward propagation time, and backpropagation time exhibit a linear relationship with batch sizes $b_i$. So for all nodes in Cluster $A$, the computing time can be expressed as:
\begin{equation}
\begin{aligned} 
    & t_{compute}^i = a_i + {P}_i, \ \forall i\in\mathcal{N}, \\
    & {a}_i = q_i b_i + s_i, \ \forall i\in\mathcal{N}, \\
    & {P}_i = k_i b_i + m_i, \ \forall i\in\mathcal{N},
\end{aligned} 
\end{equation}
where $a_i$ is the total time of parameter updating, data loading, and forward propagation, with the corresponding coefficients $q_i$ and $s_i$. ${P}_i$ is the backpropagation time, while $k_i$ and $m_i$ are coefficients related to GPU types and DL jobs. Note that within a heterogeneous GPU cluster, different GPU models exhibit varying pairs of $q_i$ and $s_i$, as well as $k_i$ and $m_i$, even when performing the same DL job. If cluster $A$ has $n$ different types of GPUs, there are $n$ different pairs of linear functions corresponding to each type of GPUs.

\subsubsection{Gradient Synchronization Time} We focus on the ring All-reduce mechanism~\cite{ringallreduce} adopted by Pytorch DistributedDataParallel~\cite{li2020pytorch}. Ring All-reduce is a synchronized communication method that starts the gradient synchronization when all nodes are ready to synchronize and ends the synchronization when all nodes finish the gradient synchronization. The gradient synchronization time $T_{comm}$ is dependent on model size (size of gradient parameters) and network status. In the scenario that the network and allocated resources are stable in a heterogeneous cluster, even though each node's network performance varies, the gradient synchronization time $T_{comm}$ is a learnable constant when we train the same job with different batch sizes.

\subsubsection{Computing and Communication Overlap} 
\label{compcommoverlap}
Modern distributed training frameworks~\cite{li2020pytorch,horovod} support the overlap between gradient computing and synchronization by separating the locally-computed gradients into buckets~\cite{li2020pytorch}. Rather than synchronizing after all nodes 
finish computing the full gradient at the end of backpropagation, each gradient bucket starts synchronization when all nodes finish computing the gradient of the same bucket. In batch processing, only the last bucket cannot overlap its synchronization with its gradient computation. So Cluster $A$'s per batch gradient synchronization time $T_{comm}$ is the sum of $T_u$, the last gradient bucket synchronization time, and $T_o$, the gradient synchronization time for all the other gradient buckets: $T_{comm}=T_{o}+T_{u}$.


For all nodes in Cluster $A$, the gradient size (model size) is determined when training starts. Different types of GPUs have the same gradient computing procedure~\cite{gpuprocess}. Even though the local batch size of each node is different, the gradient bucket will be ready for synchronization at a fixed proportion of each node's backpropagation time $P_i$. We can infer that the starting point of the gradient synchronization of node $i$ is:
\begin{equation}
     \textit{syncStart}_i = a_i + \gamma P_i,
\end{equation}
where the overlap ratio $\gamma$ is the ratio of the first bucket computing time to the total backpropagation time. The first bucket computing time is the period from the start point of the backward pass to the first bucket ready for synchronization point. The first bucket computing cannot be overlapped with the gradient synchronization. 

Although varying the local mini batch size $b_i$ changes the backpropagation time $P_i$, $T_{comm}$ is a fixed, learnable constant.
So the computing and communication overlap pattern differs when the local mini batch size varies. The first gradient bucket's ready-for-synchronization point is determined by the computing time along with $\gamma$, which is a constant that can be accurately measured through all the nodes in the cluster. The following buckets' synchronization can be blocked by the previous buckets' synchronization. To eliminate the effect of measurement error and system overhead, we only measure the first bucket's ready-for-synchronization point $\textit{syncStart}$ and assume all buckets' computing time and communication time are evenly distributed in the rest of gradient computing time and communication time. There are two possible overlap patterns of computing and communication.

\begin{figure}[h]
    \centering
    \includegraphics[width=0.45\textwidth]{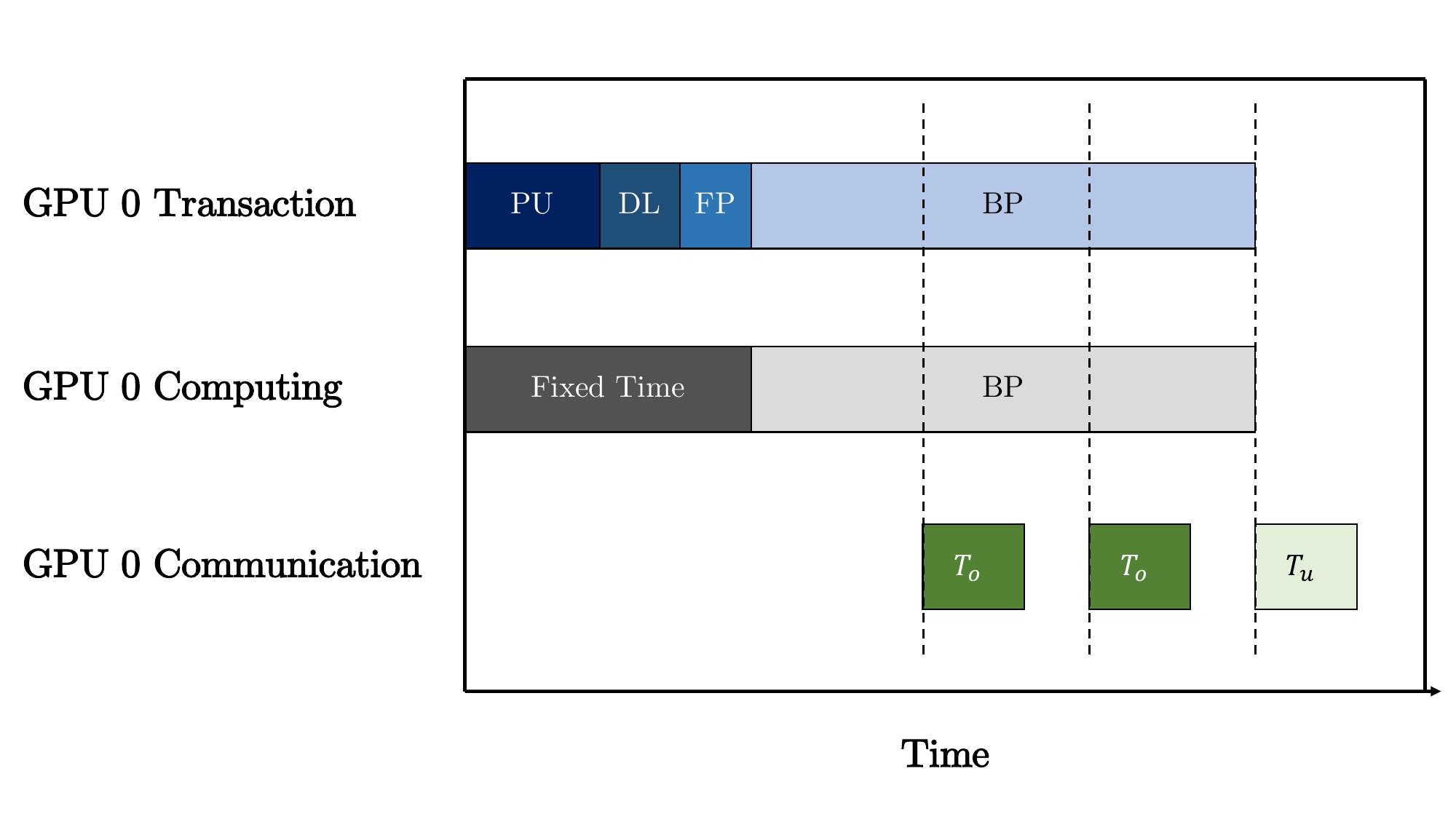}
    \caption{A node running in the computing-bottleneck situation. PU, DL, FP, and BP are the parameter update, data loading, forward propagation, and backpropagation.}
    \label{fig:computing-bottleneck}
\end{figure}

\noindent \textbf{When computing is the bottleneck.} Figure \ref{fig:computing-bottleneck} shows a computing-bottleneck node. The dashed lines show the buckets' ready-for-synchronization points. When $(1-\gamma)P_i \geq T_{o}$, node $i$'s bottleneck is the gradient computation. In this case, the gradient synchronization of each bucket finishes before the next gradient bucket is ready for synchronization, so $T_{o}$ fully overlaps with the gradient computing. The total processing time of one batch for node $i$ in cluster $A$ is:  
\begin{equation}
    T = t_{compute}^i + T_{u}.
\label{eq:compute}
\end{equation}

\begin{figure}[h]
    \centering
    \includegraphics[width=0.45\textwidth]{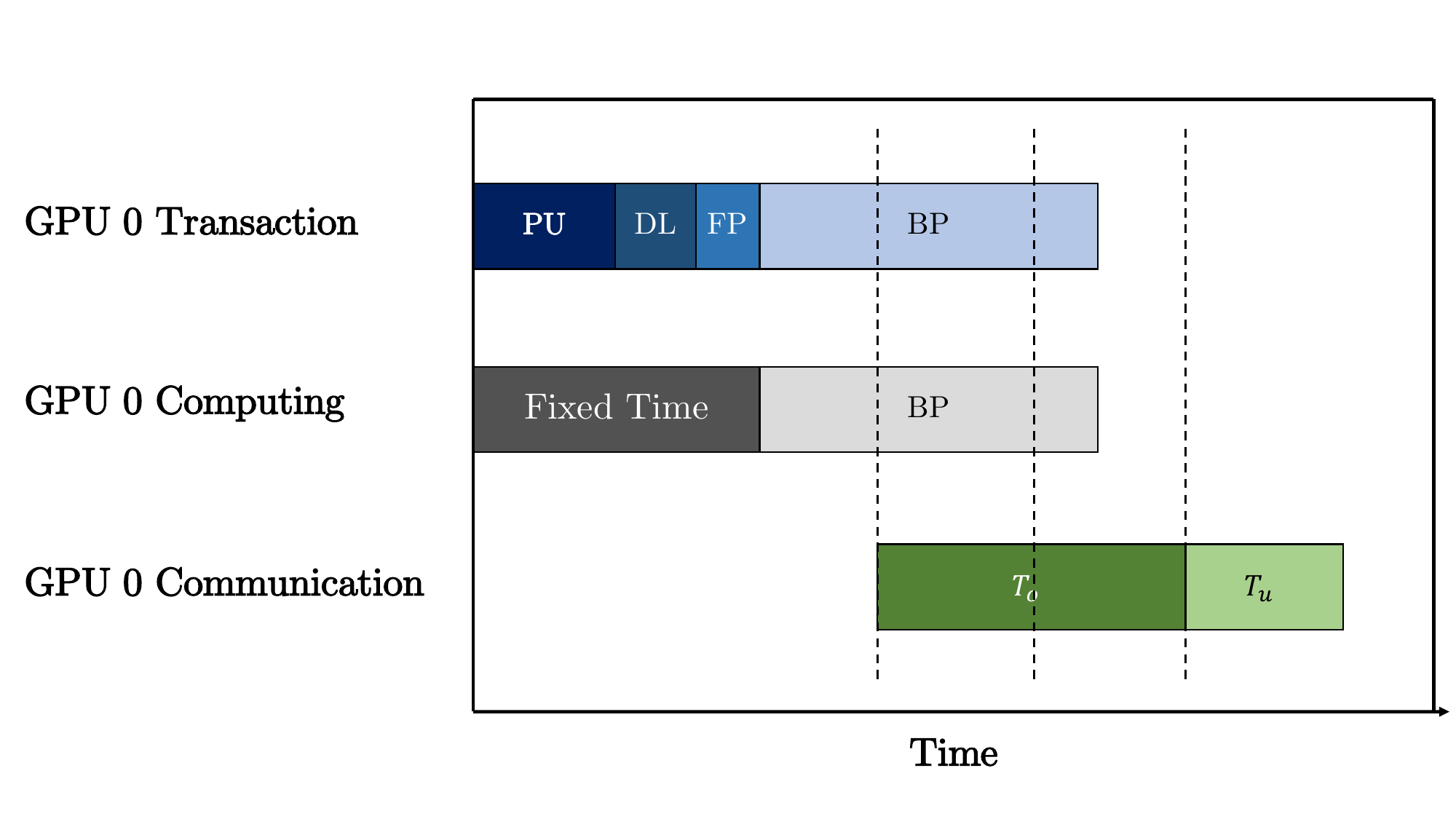}
    \caption{A communication-bottleneck node.}
    \label{fig:communication-bottleneck}
\end{figure}

\noindent \textbf{When communication is the bottleneck.} Figure \ref{fig:communication-bottleneck} shows the communication-bottleneck pattern. If $(1-\gamma)P_i < T_{o}$, node $i$'s bottleneck is the gradient communication. Although the unfinished bucket synchronization won't block gradient computing (due to the All-reduce mechanism), the synchronization of previous unfinished buckets will block the synchronization of the following bucket. In this situation, the total processing time of one batch for node $i$ is: 
\begin{equation}
    T = \textit{syncStart}_i + T_{comm}.
\label{eq:comm}
\end{equation}

\subsection{Expression of OptPerf}
\label{optferfexpress}
To predict \textit{OptPerf} for heterogeneous clusters, with learned performance models of all GPUs, first we deduce cluster $A$'s batch processing time $T$. In cluster $A$, 
\begin{equation}
T=\max\left\{\max\limits_{i\in \mathcal{N}}\{t_{compute}^i + T_{u}\},\max\limits_{i\in \mathcal{N}}\{\textit{syncStart}_i + T_{comm}\}\right\}, 
\end{equation}
where $\mathcal{N}$ is the set of nodes in Cluster $A$. Since each node's bottleneck is unknown, minimizing $T$ is a mixed integer linear programming problem. Rather than solve an NP-hard problem, we instead provide the criteria for different situations to achieve \textit{OptPerf}.

When optimizing the performance of Cluster $A$, the batch processing time $T_i$ of stragglers during training can be reduced by adjusting the local mini batch sizes for all nodes.
Intuitively, the cluster's fast nodes should be assigned larger local mini batches, while the stragglers should have smaller batches. With the computing and communication overlap patterns of each node in Section~\ref{compcommoverlap}, we first look into two special scenarios.

\smallskip

\noindent \textbf{All nodes are computing-bottleneck.} Since $(1-\gamma)P_i \geq T_{o}$, $\forall i\in\mathcal{N}$, all the nodes' performance models are Equation~\eqref{eq:compute}. \textit{OptPerf} is achieved when all nodes in cluster $A$ have the same computing time $t_{compute}$. The proof is in Appendix~\ref{app:comp_bottleneck}.

\smallskip


\noindent \textbf{All nodes are communication-bottleneck.} If $(1-\gamma)P_i < T_{o}$, $\forall i\in\mathcal{N}$, all nodes' performance models are Equation~\eqref{eq:comm}. \textit{OptPerf} will be achieved when all nodes start the first gradient bucket synchronization at the same time. The proof is in Appendix~\ref{app:comm_bottleneck}.

\smallskip


\begin{figure}[b]
    \centering
    \includegraphics[width=0.48\textwidth]{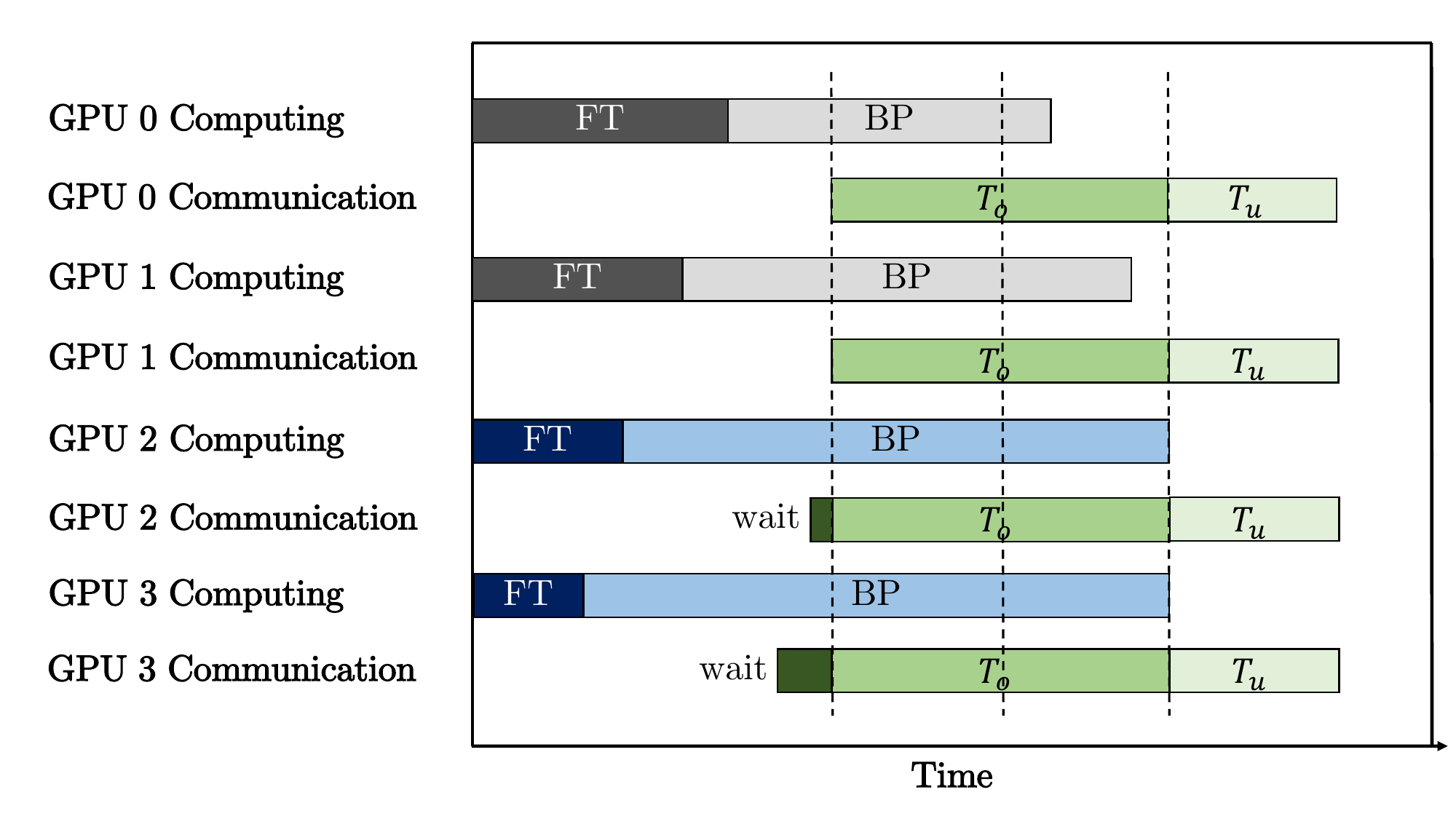}
    \caption{An example of the general case, where GPU 0 and GPU 1 are communication-bottleneck, GPU 2 and GPU 3 are computing-bottleneck. FT is the abbreviation of fixed time.}
    \label{fig:generalcase}
\end{figure}

\noindent \textbf{The general case.} In the general case, some nodes' bottlenecks are computing while others' are communication (see Figure~\ref{fig:generalcase}).

\noindent The computing-bottleneck nodes' performance models follow Equation~\eqref{eq:compute}, and the communication-bottleneck nodes' performance models follow Equation~\eqref{eq:comm}. \textit{OptPerf} is achieved when all computing-bottleneck nodes have the same computing time $t_{compute}$ and all communication-bottleneck nodes start the first bucket synchronization at the same time. Moreover, the computing and communication bottleneck nodes simultaneously get ready for the last bucket synchronization. The proof is located in Appendix~\ref{app:general_opt}. 

With \textit{OptPerf}, we can determine the optimal performance of a heterogeneous cluster with different batch sizes using the parameters measured and learned during training.

\section{System Design of Cannikin}

In this section, we give details on the workflow of Cannikin and describe how Cannikin configures the cluster before the start of each epoch, optimizes the measurement of learnable parameters, guarantees the gradient quality, and how Cannikin integrates with existing adaptive batch size training systems.

\subsection{Workflow of Cannikin}
Figure \ref{fig:workflow} shows the overview of the workflow of Cannikin. In a heterogeneous environment, after a user submits a training Job $J$ to the dynamic resource job scheduler, the job scheduler allocates a number of (possibly heterogeneous) GPUs to form Cluster $A$ to initialize Job $J$. Before the start of each epoch, the adaptive batch size engine enumerates the total batch size candidates from the batch size range~\cite{pollux}. The optimizer uses performance models learned by the analyzer to predict \textit{OptPerf} with its corresponding total batch size and \Ropt for the next training epoch, then loads each node's local mini batch based on \Ropt and starts the next training epoch. During the training epoch, each node continually collects performance metrics and learns the performance models locally. After each training epoch, the analyzer gathers the updated performance models of all nodes. 
\begin{figure}[h]
    \centering
    \includegraphics[width=0.4\textwidth]{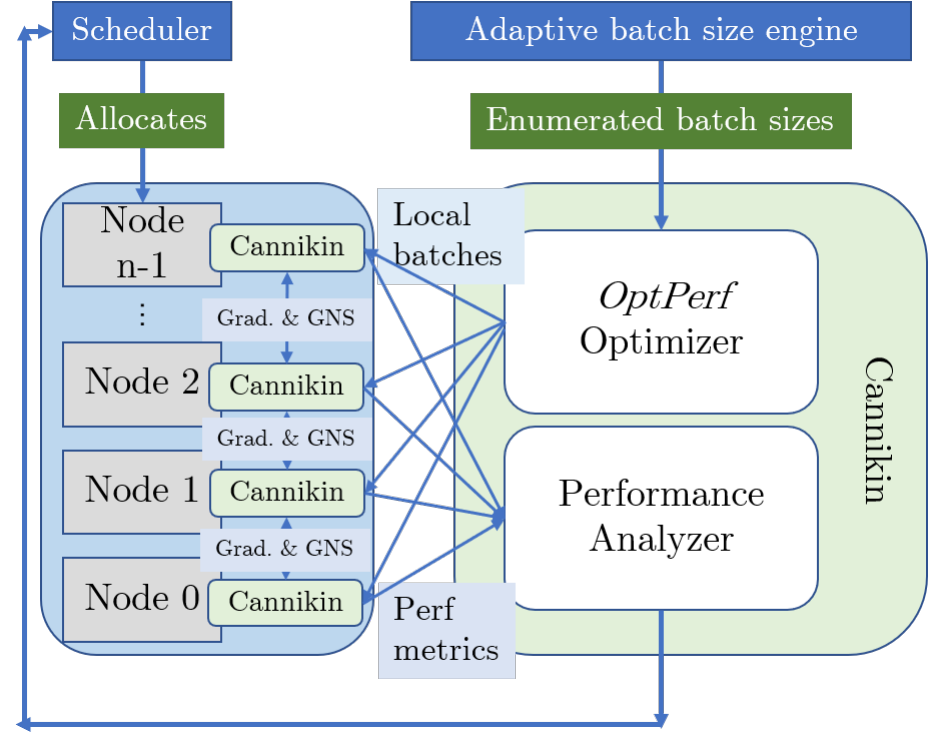}
    \caption{The overall workflow of Cannikin.}
    \label{fig:workflow}
\end{figure}


\subsection{\textit{OptPerf} Optimizer}
\label{sec:find_optperf}
Cannikin continuously learns the computing and communication models for all nodes during the training process. Despite having performance models for all nodes, the overlap state of each node in the cluster remains unknown as it is contingent on the total batch size. For example, a larger total batch size can lead to a higher number of nodes experiencing computational bottlenecks. To address this challenge, we have developed a novel search algorithm aimed at unveiling the overlap state for all nodes across the cluster.

\smallskip

\noindent \textbf{Determine the overlap state.} Given an enumerated total batch size, Cannikin uses Algorithm \ref{alg:mini batch size} to determine the overlap state, then predict \textit{OptPerf} with \Ropt for each node. 

If all nodes are computing-bottleneck or communication-bottleneck, we can use Check 1 and Check 2 to verify. 
\begin{algorithm}[t]
  \caption{Overlap state and \textit{OptPerf} configuration}
  \label{alg:mini batch size}

  {\bfseries Input:} Total batch size $B=\sum_{i=0}^{n-1}b_i$.\\
  {\bfseries Given:} $\gamma, T_o, T_u, \{a_i, k_i, m_i\}, i=0,1,\dots, n-1.$\\ $\triangleright$ Overlap ratio, two parts of communication time, computing coefficients.\\

  \textcolor{red}{/*Check 1: All nodes are computing-bottleneck*/}\\
  {\bfseries Solve} $t_{compute}=t_{compute}^0=t_{compute}^1=\dots=t_{compute}^{n-1}$.\\
  \If{$\forall (1-\gamma)P_i \geq T_{o}$}{
  \textcolor{red}{/*If all nodes are computing-bottleneck. */}\\
     \textit{OptPerf} $= t_{compute} + T_u;$\\
     {\bfseries return} \textit{OptPerf}, $b_i, i = 0, 1, \dots, n-1.$\\
     {\bfseries break}}
     \textcolor{red}{/*Check 2: If nodes are communication-bottleneck.*/}\\
     {\bfseries Solve} \textit{syncStart} $=$ \textit{syncStart}$_0=\dots=$ \textit{syncStart}$_{n-1}.$\\
    \If{$\forall (1-\gamma)P_i < T_{o}$}{
     \textcolor{red}{/*All nodes are communication-bottleneck.*/} \\
     \textit{OptPerf} $=$ \textit{syncStart} $+ T_{comm};$\\
     {\bfseries return} \textit{OptPerf}$, b_i, i = 0, 1, \dots, n-1.$\\
     {\bfseries break}}
     \textcolor{red}{/*The nodes are mixed-bottleneck.*/}\\
     
     \While{$(\exists$ \textit{syncStart}$_i>$ \textit{syncStart}$') \lor  (\exists \ t_{compute} > t_{compute}')$}{
     \textcolor{red}{/*Search the overlap state.*/} \\
     \textit{beg} $=$ \textit{outlier}$_{min};$\\
     \textit{end} $=$ \textit{outlier}$_{max};$\\
     \textit{C}$=($\textit{beg}$+$\textit{end}$)/2;$ \\ 
     \textcolor{red}{/*Node $0$ to Node $C-1$ are computing-bottleneck, Node $C$ to Node $n-1$ are comm-bottleneck.*/} \\
     {\bfseries Solve} $T_{comb}=t_{compute}' = $ \textit{syncStart}$'+T_o$.
     }     
     \textit{OptPerf}$ = T_{comb} + T_{u}.$\\
     {\bfseries return} \textit{OptPerf}$, b_i, i = 0, 1, \dots, n-1.$

\end{algorithm}
However, when nodes are mixed-bottleneck, the overhead of the enumeration method to determine each node's overlap pattern is relatively high. 
Algorithm~\ref{alg:mini batch size} addresses this issue in the following steps: If node $i$ is a computing (communication) bottleneck node in check 1 and check 2, then node $i$ is also a computing (communication) bottleneck node in the mixed-bottleneck situation. For all other outliers that have different overlap patterns in check 1 and check 2, we use a binary-search-like algorithm to determine the computing and communication bottleneck nodes. We rank all the intermediate nodes in increasing order based on the fixed processing time and then set a hypothetical bottleneck boundary node that separates the computing-bottleneck and communication-bottleneck nodes. For any overlap state, the mixed-bottleneck \textit{OptPerf} solver from Section~\ref{optferfexpress} will indicate $(\forall \ $\textit{syncStart}$\leq$ \textit{syncStart}$')$ $\land (\forall \ t_{compute} \leq t_{compute}')$ if the overlap pattern is correct. Thus we can iteratively set the middle element to be the boundary node until we find the correct overlap pattern.


Since the time complexity of checks 1 and 2 is at most $O
\left((n+1)^3\right)$ while solving linear equations~\cite{linearequations} and the time complexity of the mixed-bottleneck search algorithm is at most $O(\log n)$, Algorithm~\ref{alg:mini batch size} is $O\left((n+1)^3 \log n\right)$. In Section \ref{implementation}, we improve the time complexity of Algorithm~\ref{alg:mini batch size} to $O\left((n+1)^3\right)$.

\smallskip
\noindent \textbf{Approaching \textit{OptPerf} when no available performance models.}
As the computation time scales linearly with the local batch size of each GPU, deriving the computing time models $a_i$, $P_i$ to $b_i$ necessitates the execution of at least two distinct local batch sizes per GPU. Consequently, during the initial two epochs, there will be no available performance model to predict \textit{OptPerf}. In this scenario, we employ the inverse proportion of the sample computation time for each node to determine their respective local batch sizes for the next epoch. Assume the per sample computing time of node $i$ at the previous epoch is $t_{sample}^i=\frac{t_{compute}^i}{b_{current}^i}$, where $b_{current}^i$ is the local batch size of node $i$ in the previous epoch. The local batch size of node $i$ for the next epoch can be expressed as:
\begin{equation}
    b_{next}^i=\frac{\sum_{i\in\mathcal{N}}t_{sample}^i}{t_{sample}^i}(\sum_{i\in\mathcal{N}} \frac{\sum_{i\in\mathcal{N}}t_{sample}^i}{t_{sample}^i})^{-1}B,\label{eq:no_model}
\end{equation}


where $B$ is the total batch size for the upcoming epoch. With this method, each node can experiment with various local mini-batch sizes necessary for performance model learning, and Cannikin iteratively approaches \textit{OptPerf} when no available performance models. Note that the primary purpose of this method is to adjust each node's local batch size for performance model learning, rather than relying on the less efficient iterative method to find \textit{OptPerf}. Once the performance models are established, they are employed to predict \textit{OptPerf} before each epoch.

\subsection{Optimized Gradient Aggregation}
In homogeneous environments, the local gradient of each node can be aggregated by the cluster via averaging. This procedure guarantees each training sample has an identical weight in the global gradient after synchronization. However local gradient averaging cannot be utilized for adaptive local batch training in heterogeneous clusters due to the variety of local batch sizes. Simply averaging each node's local gradient results in over-representation of training samples from smaller local batches in the global gradient. To address this problem, LB-BSP~\cite{socc20} introduced proportional-weighted gradient aggregation. By weighting each local gradient proportionally to the local batch size, samples assigned to different nodes have identical weights in the global gradient. Cannikin computes the global gradient $g$ using:
\begin{equation}
    g = \sum_{i\in\mathcal{N}}r_i  g_i,\label{eq:grad_agg}
\end{equation}
where $g_i$ is the local gradient in Equation (\ref{eq:localgradient}) and $r_i$ is the local mini batch ratio of node $i$. For \emph{i.i.d.} data, $g$ is equivalent to the averaged gradients in homogeneous environments.


\subsection{Gradient Noise Scale in Heterogeneous Environment}
Adaptive batch size training uses the gradient noise scale (GNS)~\cite{gns} to model the convergence efficiency (statistical efficiency). The GNS, $\mathcal{B}_{noise}$, measures how large the gradient is compared to its variance: $\mathcal{B}_{noise}=\text{tr}(\Sigma)/|G|^2$,
where $\Sigma$ is the covariance matrix and $G$ is the noiseless true gradient. Since $\text{tr}(\Sigma)$ and $|G|^2$ are not available in practice, standard methods instead rely on good estimators of these two quantities. Previous work has only considered how to estimate $\text{tr}(\Sigma)$ and $|G|^2$ (and thus $\mathcal{B}_{noise}$) in homogeneous clusters.

To compute the GNS, we first construct local estimates $\mathcal{G}_i$ and $\mathcal{S}_i$ of $|G|^2$ and $\text{tr}(\Sigma)$ for each node $i$:
\begin{align}
    \mathcal{G}_i &= \frac{1}{B-b_i} (B|g|^2 - b_i|g_i|^2), \quad \mathcal{S}_i = \frac{b_i B}{B-b_i} (|g_i|^2-|g|^2)\label{eq:gi}
\end{align}
where the estimates incorporate local and global gradient information. For any batch of size $b$, the expected gradient norm $\mathbb{E}[|g_{est}|^2]$ satisfies the equality $\mathbb{E}[|g_{est}|^2]=|G|^2 + \frac{1}{b}\text{tr}(\Sigma)$~\cite{gns}. Using this equation, we can prove that $\mathcal{G}_i$ and $\mathcal{S}_i$ are unbiased estimators of $|G|^2$ and $\text{tr}(\Sigma)$, respectively. Aggregating the local estimates $\mathcal{G}_i$ and $\mathcal{S}_i$ across all nodes can provide high quality, unbiased estimates of $|G|^2$ and $\text{tr}(\Sigma)$ that have improved, lower variance. The variance of these estimators is crucial since the standard ratio estimator used for the GNS is inherently biased~\cite{gns}.

In homogeneous clusters, it is optimal to separately aggregate $\mathcal{G}_i$ and $\mathcal{S}_i$ via averaging. However in Lemma~\ref{lemma:gi_and_si_var}, we prove that the variance of both $\mathcal{G}_i$ and $\mathcal{S}_i$ depend on the local mini batch size. Furthermore, the local estimates of $\text{tr}(\Sigma)$ and $|G|^2$ for different nodes are correlated via dependence on $|g|^2$. As a result, aggregating $\mathcal{G}_i$ and $\mathcal{S}_i$ is more challenging for heterogeneous clusters. The following theorem states the optimal weighted combination of the local estimators, with the proof located in Appendix~\ref{app:gns}.
\begin{theorem}
    \label{thm:opt_g_and_s}
    $\mathcal{G} = \sum_{i\in\mathcal{N}}w^{\mathcal{G}}_i \mathcal{G}_i$ and $\mathcal{S} = \sum_{i\in\mathcal{N}}w^{\mathcal{S}}_i \mathcal{S}_i$ are minimum variance, unbiased linear estimators of $|G|^2$ and $\text{tr}(\Sigma)$ when:
    \begin{align}
        \mathbf{w}^{\mathcal{G}} = \frac{\mathbf{1}^T A_{\mathcal{G}}^{-1}}{\mathbf{1}^T A_{\mathcal{G}}^{-1}\mathbf{1}},\quad \quad \mathbf{w}^{\mathcal{S}} = \frac{\mathbf{1}^T A_{\mathcal{S}}^{-1}}{\mathbf{1}^T A_{\mathcal{S}}^{-1}\mathbf{1}},\label{eq:opt_g_and_s}
    \end{align}
    where $\mathbf{1}$ is an $n$-dimensional column vector of ones and both $A_{\mathcal{G}}$ and $A_{\mathcal{S}}$
    are $n\times n$ matrices with respective entries $a_{\mathcal{G}}(i,j)$ and $a_{\mathcal{S}}(i,j)$:
    \begin{equation*}
        a_{\mathcal{G}}(i,i) = \dfrac{B+2b_i}{B^2-Bb_i}, \quad a_{\mathcal{G}}(i,j) = \dfrac{B^2-b_i^2-b_j^2}{B(B-b_i)(B-b_j)} \text{ for } i\neq j
    \end{equation*}
    \begin{equation*}
        a_{\mathcal{S}}(i,i) = \dfrac{Bb_i}{B-b_i}, \quad a_{\mathcal{S}}(i,j) = \dfrac{b_i b_j(B-b_i-b_j)}{(B-b_i)(B-b_j)} \text{ for } i\neq j
    \end{equation*}
\end{theorem}

Cannikin delivers a novel method to estimate the GNS $\mathcal{B}_{noise}$ in heterogeneous clusters. First, each node estimates the sum of the variances of the individual gradient components and the global norm of the gradient using~\eqref{eq:gi}. We optimally aggregate the local estimates $\mathcal{G}_i$ and $\mathcal{S}_i$ using~\eqref{eq:opt_g_and_s}, and then take the ratio of the resulting terms to get the global gradient noise scale $\mathcal{B}_{noise}=\mathcal{S}/\mathcal{G}$. Despite the added challenge of heterogeneity, Figure~\ref{fig:gns} shows Cannikin's convergence is comparable to the homogeneous baseline AdaptDL with the same training epochs, which means the larger batch size chosen by Cannikin won't harm the convergence efficiency.
\begin{figure}[h]
    \centering
    \includegraphics[width=0.48\textwidth]{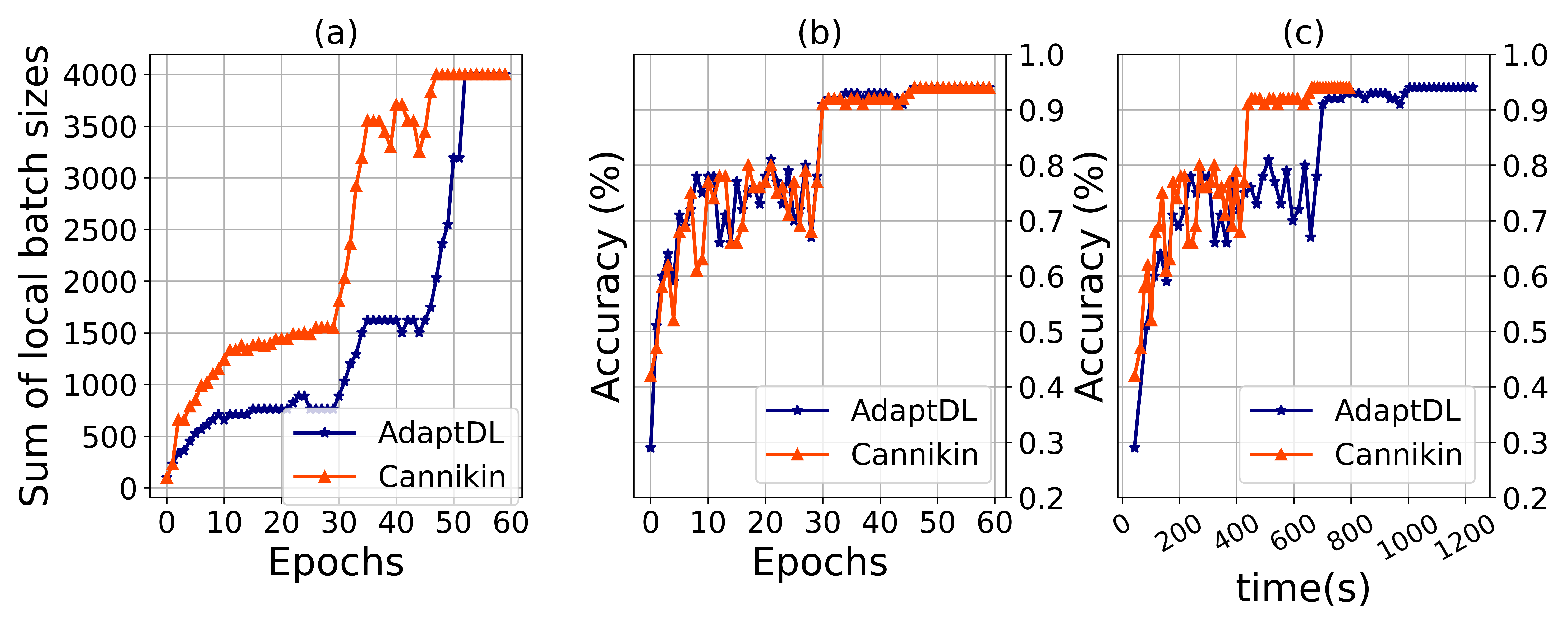}
    \caption{(a) shows the total batch size selected by Cannikin and AdaptDL in a heterogeneous cluster for the training of CIFAR10. In most epochs, Cannikin trains with a larger batch size compared with AdaptDL due to the throughput improvement. (b) shows with the same training epochs, Cannikin achieves the same accuracy compared with AdaptDL. (c) shows the accuracy in the convergence process of Cannikin and AdaptDL.}
    \label{fig:gns}
\end{figure}

\subsection{Implementation}
\label{implementation}
Cannikin is implemented as a PyTorch library based on AdaptDL~\cite{pollux} that can be imported into DNN training scripts. Cannikin introduces the \texttt{HeteroDataLoader} class, which unevenly loads local mini batches to each node based on the \textit{OptPerf} prediction. The implementation of Cannikin addresses the following concerns to improve efficiency.

\smallskip

\noindent \textbf{Parameter learning.}
During each epoch, Cannikin collects the backpropagation time ($P_i$) and the total time of data loading, optimization steps, forward propagation($a_i$) for each local batch size. With the collected data from two epochs using different local mini-batch sizes, each node can construct the computing time model to the local mini batch size by solving linear equations. In the subsequent epochs, employing more different local batch sizes, allows for the refinement of the computation time model, making it increasingly accurate.

When it comes to learning the communication time ($T_{comm}$) and $\gamma$, it's important to note that $T_{comm}$ and $\gamma$ remain constant across different local and total batch sizes. Cannikin collects the overlap ratio ($\gamma$) as well as communication times ($T_o$ and $T_u$) for each node in the cluster. We proceed to optimize the learning of $\gamma$ and $T_{comm}$ as follows.

\noindent \textbf{Optimized parameter measurement in the cluster.}
\label{optimize measurement}
Cannikin collects observation-based parameters such as overlap ratio $\gamma$ and communication time $T$ from all nodes in the cluster. However since different nodes can have different levels of noise in their measurement of $\gamma$, simply averaging across all nodes causes significant error. 
Figure~\ref{fig:overlapratio} shows the measured $\gamma$ for ResNet18 training on Cifar10 on different GPUs. The stochastic nature of distributed DNNs results in randomness when measuring $\gamma$~\cite{zeus}, which can lead to error when learning $\gamma$. 
\begin{figure}[h]
    \centering
    \includegraphics[width=0.45\textwidth]{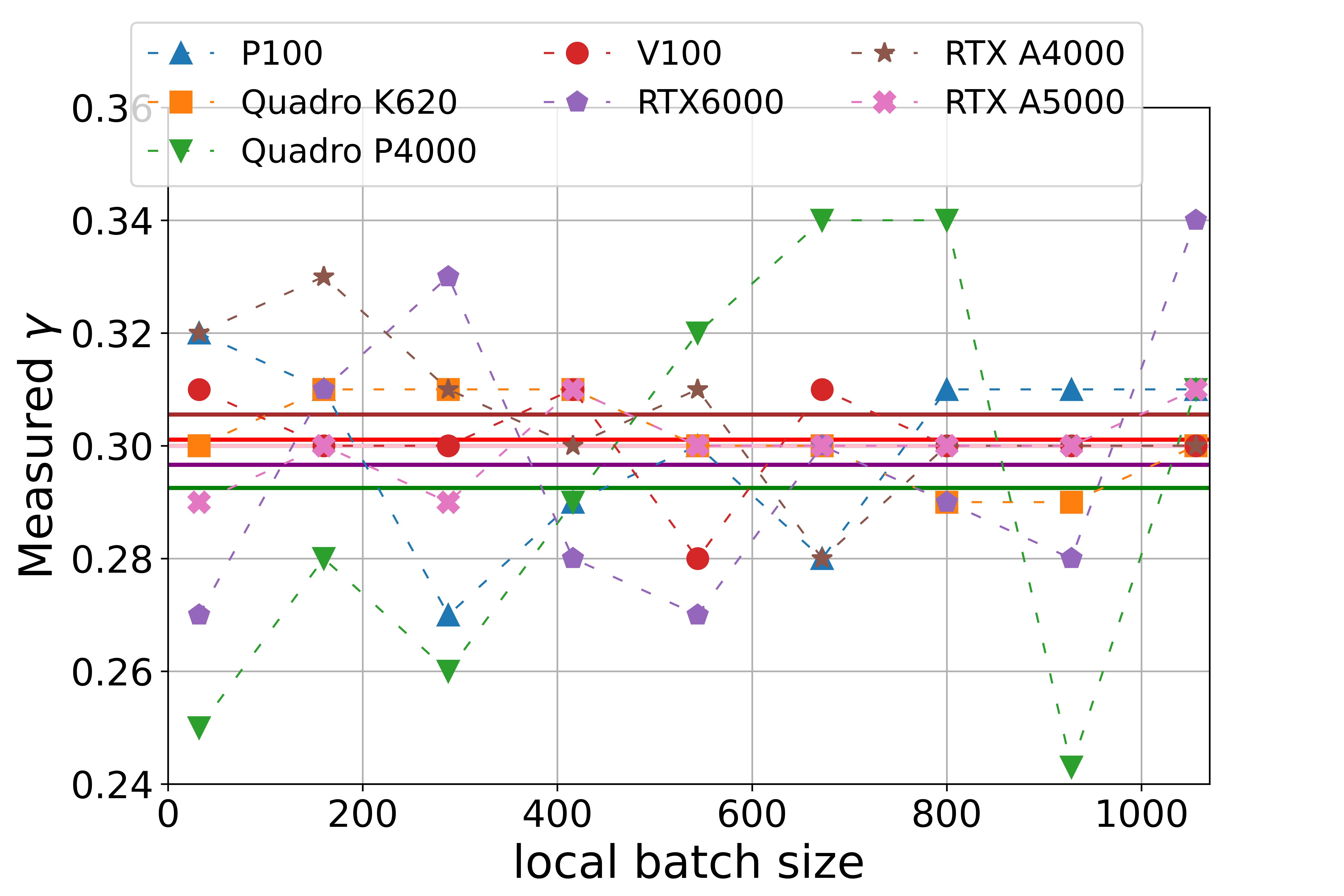}
    \caption{The measured overlap ratio $\mathbf{\gamma}$ of different local mini batch sizes validated on different types of GPUs~\cite{cloudlab}.}
    \label{fig:overlapratio}
\end{figure}
We adopt the inverse-variance weighting~\cite{inversevirance} to adjust the weighting of each node's measurement of $\gamma$:
\begin{equation}
    \gamma = \sum_{i\in\mathcal{N}}\frac{\gamma_i}{\hat{\sigma}_{\gamma_i}^2} \left(\sum_{i\in\mathcal{N}}\frac{1}{\hat{\sigma}_{\gamma_i}^2}\right)^{-1} ,
\end{equation}
where $\gamma_i$ is the overlap ratio estimation of node $i$ and $\hat{\sigma}_{\gamma_i}^2$ is the sample variance of $\gamma_i$. This estimation of $\gamma$ is optimal when observation errors are uncorrelated across nodes.

The communication time $T$ is also a fixed value across all nodes because of the ring All-reduce mechanism~\cite{ringallreduce}. However, each node reports different $T_i$ during training in heterogeneous clusters because of the wait-for-synchronization time of some nodes. In Cannikin's implementation, we use $T = \min_{i\in\mathcal{N}}\{T_i\}$ to eliminate the error in communication time measurement because this value typically corresponds to the slowest node that does not need to wait for any other node for synchronization.

\smallskip

\noindent \textbf{Total batch size selection.} Although the search algorithm is efficient in finding the overlap pattern, the overhead can be significant if we determine the overlap pattern for each total batch size candidate determined by the adaptive batch size engine~\cite{pollux} before every epoch. Cannikin instead calculates $\textit{OptPerf}_{init}$ for all batch size candidates after the initial epoch. In the upcoming epochs, since \textit{OptPerf} is unrelated to the training progress, Cannikin uses $\textit{OptPerf}_{init}$ and the updated gradient noise scale to choose the total batch size. Then Cannikin determines \textit{OptPerf} along with \Ropt according to the updated performance metrics. 
If the overlap pattern has changed from the initial pattern, Cannikin will start over to determine the pattern for each candidate again to choose the total batch size. Otherwise Cannikin will update $\textit{OptPerf}_{init}$ for the corresponding total batch size candidate. With this strategy, in most epochs Cannikin only needs to determine \textit{OptPerf} for one total batch size.

\smallskip

\noindent \textbf{Overlap state searching.} In the initialization epoch, Cannikin goes through all the total batch size candidates and calculates \textit{OptPerf} for each candidate. When the total batch size increases, more cluster nodes will be computing-bottleneck nodes. Hence in the total batch size enumeration from small to large in sequence, the search starting point of an enumerated candidate is the overlap pattern of the previous one. In the following epochs, the search starting point of an enumerated candidate is its overlap state in $\textit{OptPerf}_{init}$.

\smallskip

\noindent \textbf{Integer batch sizes.}
The batch size of each node is required to be an integer. To avoid the overhead of integer programming, Cannikin instead computes the optimal $b_i$ without enforcing the integer constraint and then rounds any decimal local mini batch sizes. In our evaluation, the error caused by the approximation is insignificant.

\section{Evaluation}
We evaluate the effectiveness of Cannikin in optimal distributed DNN training using convergence time, batch processing time, prediction accuracy, and the system overhead. Key results are as follows:
\begin{itemize}
    \item Cannikin reduced the overall convergence time by up to $85\%$ and $52\%$ in heterogeneous clusters compared with PyTorch and the adaptive batch size training system AdaptDL.
    \item Compared to the state-of-art data parallel distributed deep learning strategies for heterogeneous clusters, Cannikin reduces the batch processing time by up to $18\%$ compared to the baseline.
    \item Cannikin predicted \textit{OptPerf} for heterogeneous clusters within $7\%$ error with low overhead less than $4\%$.
\end{itemize}

\subsection{Experimental Setup}

\noindent \textbf{Testbed.} We conduct our experiments in two different clusters: cluster $A$ and cluster $B$. Cluster $A$ is a heterogeneous 3-node cluster with different types of NVIDIA GPUs specified in Table~\ref{tab:clustera}. Cluster $B$ is a heterogeneous 10-server cluster consisting of $16$ GPUs, as detailed in the table~\ref{tab:clusterb}. Note that in Cluster $B$, each GPU is a node for data-parallelism distributed DL training.

\begin{table}[h]
  \caption{Hardware specification of cluster \emph{A} in evaluation}
  \centering
    \footnotesize
  \begin{tabular}{cccccc}
    \toprule
    \small
    \textbf{Node} & \textbf{Node} &\textbf{GPU}  &\textbf{GPU} & \textbf{Main} & \multirow{2}* {\textbf{CPU}}  \\ 
    \textbf{type}& \textbf{count}& \textbf{model}& \textbf{Count}& \textbf{Memory} &  \\  
    \midrule
    a5000 & 1 & RTX A5000 &1 &32GB &i9-10980XE \\
    a4000 & 1 & RTX A4000 &1 &32GB&Xeon W-2255 \\
    p4000 & 1 & Quadro P4000 &1 &32GB &Xeon W-2102 \\

  \bottomrule
  \end{tabular}

  \label{tab:clustera}
\end{table}

\begin{table}[h]
  \caption{Hardware specification of cluster \emph{B} in evaluation}
  \centering
   \footnotesize
  \begin{tabular}{cccccc}
    \toprule
    \small
    \textbf{Node} & \textbf{Node} &\textbf{GPU}  &\textbf{GPU} & \textbf{Main} & \multirow{2}* {\textbf{CPU}}  \\ 
    \textbf{type}& \textbf{count}& \textbf{model}& \textbf{Count}& \textbf{Memory} &  \\  
    \midrule
    a100 & 1 & A100 &4 &512GB &Xeon Plati. 8380*2 \\
    v100 & 1 & V100 &4 &128GB&Xeon Gold 6230*2 \\
    rtx & 8 & RTX6000 &1 &192GB &Xeon Gold 6126*2 \\

  \bottomrule
  \end{tabular}
  \label{tab:clusterb}
\end{table}

\begin{table*}[t]
    \caption{The training models,  with the datasets ImageNet, CIFAR-10, LibriSpeech, SQuAD, and MovieLens. The model size (parameters in the model), optimizer, learning rate scaler, original batch size $\mathbf{B_0}$, and the target metrics used for Cannikin's evaluation. Note that we use fine-tuning BERT for evaluation.}
  \centering
  \begin{tabular}{ccccccccc}
    \toprule
        \textbf{Task} & \textbf{Dataset} & \textbf{Model} & \textbf{Size} & \textbf{Optimizer} & \textbf{LR scaler} & $\mathbf{B_0}$ & \textbf{Target} \\ \midrule 
        
        Image Classification & ImageNet~\cite{imagenet} & ResNet-50~\cite{resnet}&  25.6M & SGD & Adascale & 100 & 75\% Top1 acc. \\ 
        
        Image Classification & CIFAR-10~\cite{cifar} & ResNet-18~\cite{resnet}&  11M  & SGD & Adascale & 64 & 94\% Top1 acc. \\
        
        Speech Recognition & LibriSpeech~\cite{libri} &  DeepSpeech2~\cite{deepspeech}&  52M  & SGD & Adascale & 12 & WER = 40.0\% \\
        
        Question Answering & SQuAD~\cite{squad} & BERT~\cite{bert}&  110M & AdamW & Square-Root & 9 & F1 = 88\% \\ 
        
        Recommendation & MovieLens~\cite{movielens} & NeuMF~\cite{neumf} &  5.2M & Adam & Square-Root & 64 & Hit rate = 69\% \\ \bottomrule
  \end{tabular}
  \label{tab:workloads}
\end{table*}

\smallskip

\noindent \textbf{Workloads.} Evaluated workloads are listed in Table~\ref{tab:workloads}. The range of batch sizes is determined by each GPU's memory; the initial batch size is relatively small~\cite{pollux} and configured by users. For the optimizer, learning rate scaler, and target metrics choices, we adopt the canonical setting for each training task, with the philosophy of testing different models and optimizers on applications of various sizes.

\smallskip

\noindent \textbf{Baselines.} We evaluate Cannikin by comparing it with the state-of-art adaptive batch size training system, data-parallelism heterogeneous distributed DL training system, and PyTorch DDP:
\begin{itemize}
    \item AdaptDL~\cite{pollux}: The state-of-art adapted distributed DNN training system for homogeneous clusters.
    \item LB-BSP~\cite{socc20}: LB-BSP is a data-parallelism distributed training system for heterogeneous GPU clusters, which recurrently tune each node's local mini batch size for efficient model training. We set step size $\Delta=5$ in our experiments, which is identical to the original paper.
    \item PyTorch DistributedDataParallel~\cite{li2020pytorch}: Pytorch DDP is one of the most efficient distributed training libraries for homogeneous clusters.
\end{itemize}

\subsection{Performance with Heterogeneous GPUs}

\subsubsection{Overall convergence performance}
For the overall convergence performance evaluation, we compared Cannikin with the baselines in cluster $B$. Figure~\ref{fig:convergeeaxmple} shows the convergence processes of example tasks Cifar10 and Imagenet training in cluster $B$. Due to the weighted gradient aggregation, each gradient descent step of Cannikin is equivalent to the homogeneous gradient descent given the same total batch size, This equivalence guarantees that the convergence is not compromised. With this precondition, Cannikin achieves a convergence speed up from throughput improvement and the improved prediction of optimal total batch size in heterogeneous clusters, thus increasing the adaptive training system's goodput. Our results show Cannikin reduces the convergence time by $52\%$ and $29\%$ for CIFAR-10 and ImageNet compared with AdaptDL.

\begin{figure}[h]
    \centering
    \includegraphics[width=0.48\textwidth]{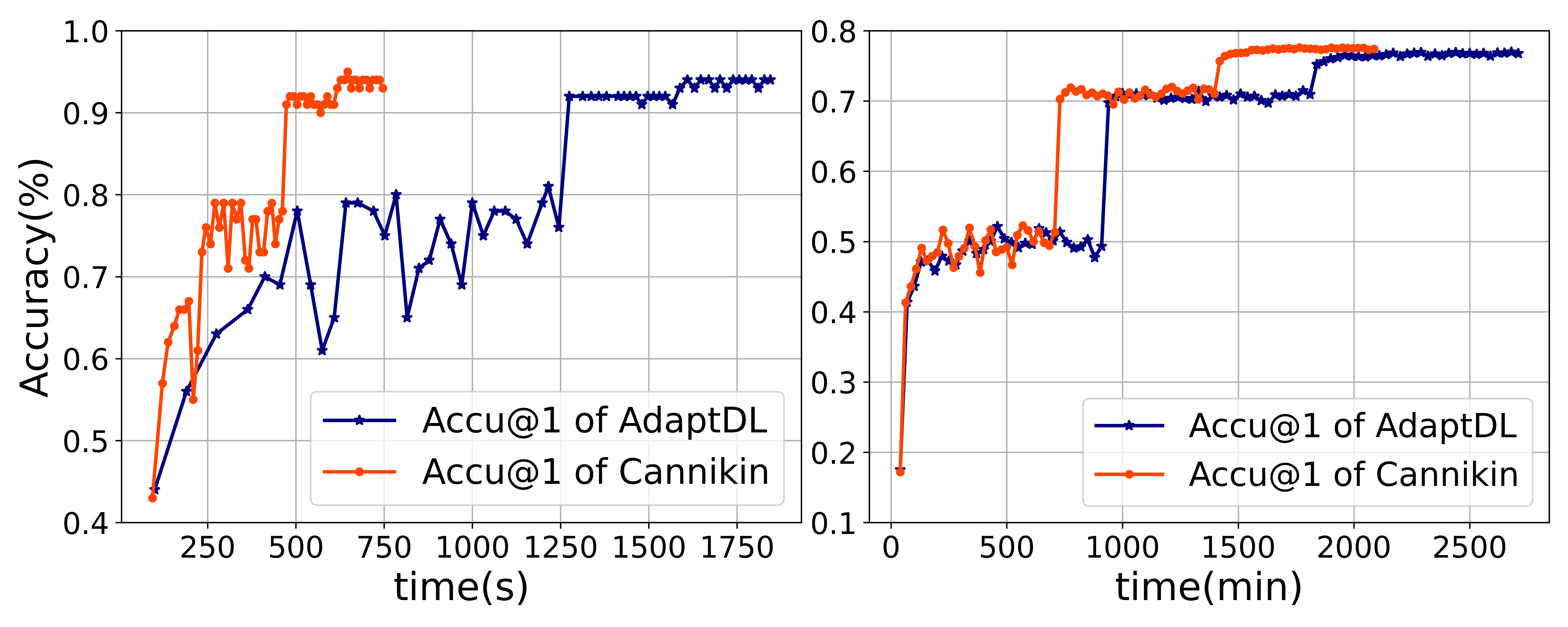}
    \caption{Convergence process of ResNet-18 on CIFAR-10 (left) and ResNet-50 on ImageNet (right).}
    \label{fig:convergeeaxmple}
\end{figure}

\begin{figure}[h]
    \centering
    \includegraphics[width=0.5\textwidth]{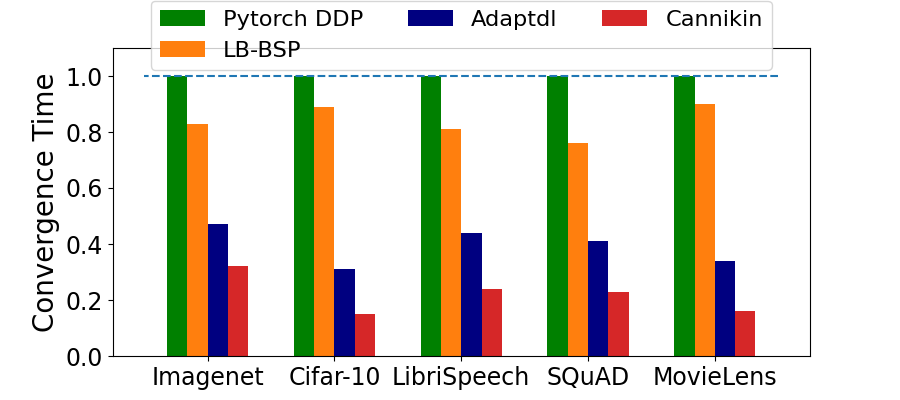}
    \caption{Normalized convergence time of all training tasks.}
    \label{fig:realcluster}
\end{figure}
The normalized overall convergence time for each evaluated workload is depicted in Figure~\ref{fig:realcluster}. In the case of PyTorch DDP, which trains deep learning models with fixed total batch size and distributes local batch sizes evenly across heterogeneous clusters, the speedup achieved by Cannikin is primarily attributed to its optimized prediction of total/local batch sizes during training. Unlike PyTorch DDP, which uses fixed batch sizes, AdaptDL evenly distributes local batch sizes across the cluster and predicts the optimal total batch size in homogeneous environments. In the context of AdaptDL, the speedup observed with Cannikin results from the optimized selection of local batch sizes to maximize the utilization of heterogeneous GPUs and the improved prediction of total batch sizes in heterogeneous environments. LB-BSP iteratively tunes local batch sizes for each GPU within heterogeneous clusters which improves the utilization of the heterogeneous GPUs. However, LB-BSP only supports the DL training with a fixed total batch size and LB-BSP hasn't consider the communication and communication overlap. The speedup with Cannikin compared to LB-BSP primarily arises from the faster determination of the optimal local batch sizes considering communication and computing overlap and optimized total batch size selection during training. The results indicate Cannikin significantly enhances in overall convergence time to achieve the target accuracy, with improvements of up to $85\%$, $52\%$, and $82\%$ compared to PyTorch DDP, AdaptDL, and LB-BSP respectively.

\subsubsection{Batch processing time for heterogeneous clusters}
We evaluate Cannikin using two methodologies for batch processing time. The first is the fixed total batch size training, i.e. classical DNN training. The second is the adaptive batch size situation when the total batch size of the cluster varies in each training epoch. Since AdaptDL's batch processing time in heterogeneous clusters is equivalent to Pytorch DDP, we don't consider AdaptDL in this section.

\smallskip

\noindent \textbf{With fixed batch size.}
We fix the total batch size of the cluster and each node's optimal local mini batch size ratio \Ropt. Figure~\ref{fig:frominit} shows an example of Cannikin and LB-BSP training ResNet-50 with ImageNet in cluster $A$. Given the total batch size of 128, Cannikin and LB-BSP initialize training by evenly assigning local batch size for each node. Cannikin approach \textit{OptPerf} as early as the third epoch, because Cannikin requires two epochs to learn the performance models discussed in Section~\ref{sec:find_optperf}. However, LB-BSP requires more than ten epochs to reach its best performance. 

\begin{figure}[h]
    \centering
    \includegraphics[width=0.5\textwidth]{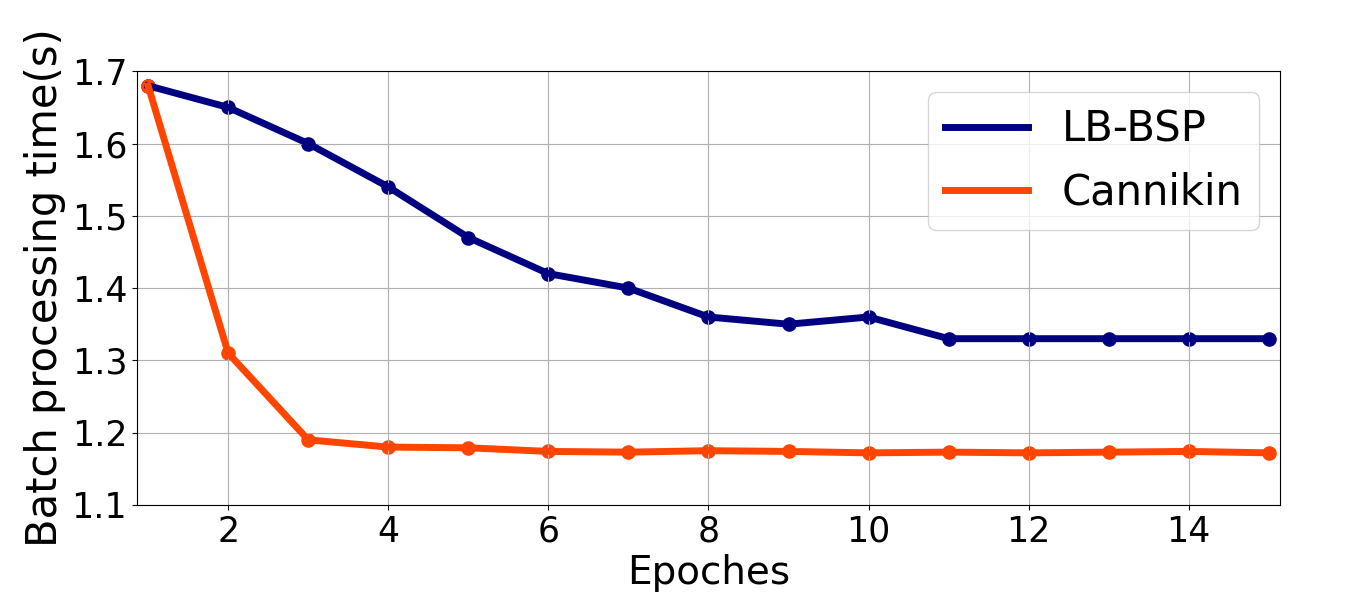}
    \caption{Cluster $\mathbf{A}$'s batch size processing time when training ImageNet from evenly assigned local mini batch size initialization given fixed total batch size 128. }
    \label{fig:frominit}
\end{figure}
\begin{figure*}[t]
    \centering
    \includegraphics[width=1\textwidth]{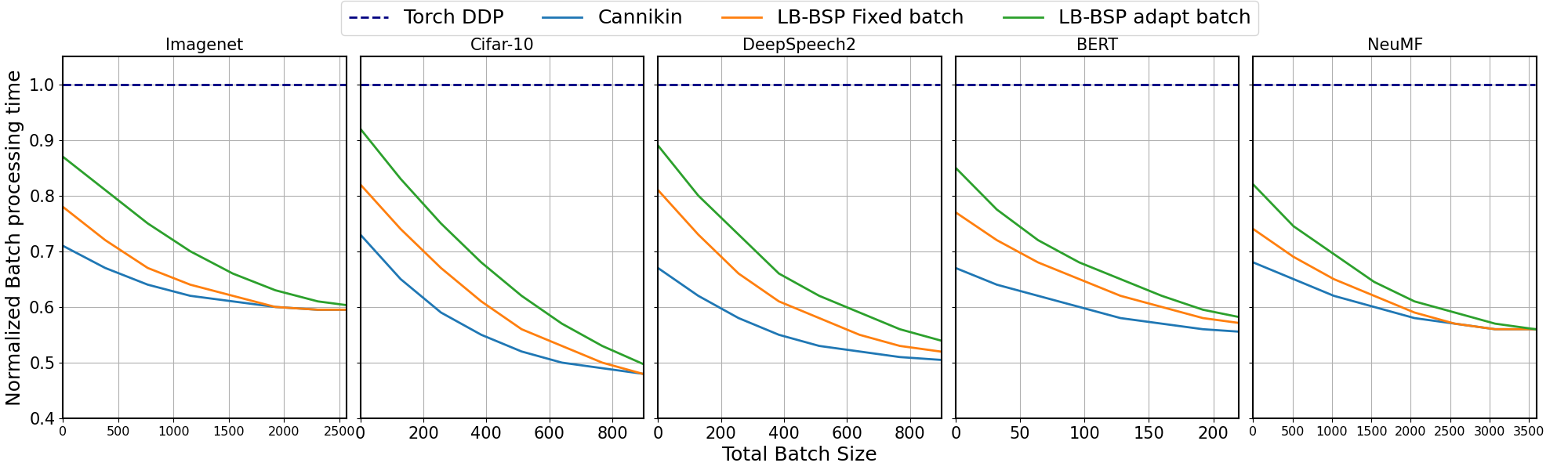}
    \caption{Each evaluation task's normalized batch processing time to the total batch size. The normalized performances of LB-BSP in fixed batch size situations and adapted batch size situations will be approached when the batch size is large enough. Because when the batch size is large, the optimal local mini batch size ratio to the total batch size of each node will approach a constant. This is also why when the batch size is large, Cannikin's normalized performance speedup will approach a constant.}
    \label{fig:batchprocess}
\end{figure*}

Assume Cannikin and each compared method have reached their best batch processing time for the given batch size. Figure \ref{fig:batchprocess} shows \textit{OptPerf} of cluster $B$ compared with the baselines. We can observe that \textit{OptPerf} is at most $18\%$ faster compared with LB-BSP, and up to 53\% faster than the Pytorch DDP. Note that when the batch size is large enough, all nodes become computing bottleneck. The \textit{OptPerf} will be achieved when all nodes have the same computing time $T_{compute}$. The performance of LB-BSP will approach \textit{OptPerf} because asymptotically the two have the same target that all nodes have the same $T_{compute}$.


\smallskip

\noindent \textbf{With adaptive batch size.}
Since only data-parallelism distributed systems are sensitive to batch size changes, we evaluate Cannikin with LB-BSP for the adaptive batch size situation. Assuming Cannikin and LB-BSP have already achieved their best performance for the previous batch size, now given a new batch size that is 10\% of the total batch size range larger than the previous one, the batch processing time of LS-BSP will become sub-optimal because the \Ropt has changed. In the meantime, Cannikin can still accurately predict the \textit{OptPerf} for the newly assigned batch size, just like the fixed batch size situation. Figure \ref{fig:batchprocess} shows the batch processing time of LS-BSP in cluster $B$ for adaptive batch size training.

\subsection{\textit{OptPerf} Prediction}

In cluster $A$, we evaluate the prediction of \textit{OptPerf} with and without inverse variance weighting in measurement compared to the manually tuned \textit{OptPerf}. Results show that without inverse variance weighting, the maximum error of \textit{OptPerf} prediction can reach up to 21\%. With the inverse variance weighting method introduced in Section~\ref{implementation}, Cannikin's prediction of \textit{OptPerf} in small and medium models like NeuMF, ResNet-18, and ResNet-50 have a maximum 3\% error. For larger models like BERT and DeepSpeech2, larger model sizes lead to more gradient buckets to synchronize, which increases the probability of contingency in gradient synchronization, so the maximum error in the prediction of \textit{OptPerf} is 7\% in the batch size range. However, Cannikin trains with varying batch sizes, so the maximum 7\% error of the \textit{OptPerf} prediction is only used in a fraction of the entire training process.

\subsection{Overhead and Scalability of Cannikin}
Table \ref{tab:overhead} shows the overhead of Cannikin for each task we deployed in the large-scale test cluster $B$. The overhead encompasses the time required to evaluate each candidate total batch size alongside its corresponding \textit{OptPerf}, as well as the configuration time for each node's local batch size and local training data index. For all the medium and large applications, the configuring time for \textit{OptPerf} of Cannikin before each epoch is much less than $1\%$ of the total epoch training time across all candidate batch sizes in the range specified by~\cite{pollux}, which is insignificant for the entire training process. For small applications like CIFAR-10 and MovieLens, the overhead of Cannikin will reach up to $9\%$ and $12\%$ when the system runs with batch sizes around the upper limit of the batch size range. However, during the training progress, the system will use the batch sizes near the upper limit only when the model almost converges. The period using the batch sizes near the upper limit for training is a minority part of training time. Considering the entire training progress, the overheads of CIFAR-10 and MovieLens are $2.7\%$ and $3.9\%$.

\begin{table}[h]
\caption{The maximum overhead for an epoch and the overall overhead for the complete training process of Cannikin.}\label{tab:overhead}
  \centering
 \small
  \begin{tabular}{cccc}
    \toprule
\multirow{2}*{\textbf{Dataset}} & \multirow{2}*{\textbf{Model}} & \textbf{Max} & \textbf{Overall} \\ 
& & \textbf{Overhead} & \textbf{Overhead} \\\midrule

ImageNet & ResNet-50 & $\ll 1\%$ & $\ll 1\%$\\ 

LibriSpeech &  DeepSpeech2& $\ll 1\%$ & $\ll 1\%$\\ 

SQuAD & BERT& $\ll 1\%$ & $\ll 1\%$\\ 

CIFAR-10 & ResNet-18 & $9\%$ & $2.7\%$\\ 

MovieLens & NeuMF & $12\%$ & $3.9\%$\\ \bottomrule

  \end{tabular}
  \label{tab:overhead}
\end{table}

\section{Discussion}

\smallskip

\noindent \textbf{Impact of varying heterogeneity.}
The performance improvement compared with the baseline depends on the heterogeneity of the cluster. Generally speaking, a cluster with more heterogeneity will get more benefits from Cannikin. In homogeneous clusters, the performance of Cannikin is identical to AdaptDL. In Cluster $B$, the fastest GPU A100 is about $3.42$ times faster compared with RTX6000 which is the slowest GPU. The degree of heterogeneity we evaluated in this paper generally exists in today's computing platform. As shown in Table~\ref{tab:gpus}, after two years, the H100 GPU is faster than A100 by more than $4$ times. We know that the A100 GPUs are not outdated and are still prevalent within the majority of deep learning clusters.

\smallskip

\noindent \textbf{Potentials with Sharing-caused heterogeneity.}
The heterogeneity can arise not only from hardware differences but also from resource sharing. Recent studies~\cite{gpu_share, antman} have introduced GPU-sharing mechanisms that enable the sharing of a single GPU's resources among multiple instances. In this context, even when the same GPU type is present within a cluster, the resources available at each node can still exhibit heterogeneity during distributed training.

We create Cluster $C$, a 16-node homogeneous cluster in Chameleon Cloud~\cite{chameleon}. Each node is equipped with one NVIDIA RTX6000 GPU. We use the container's constraint to construct the heterogeneous environment. We adopt docker containers~\cite{docker} for cluster $C$ to configure the heterogeneous environment. In each node, we start two docker containers. The first container runs Cannikin distributed training workloads, and the second docker container runs a local dummy GPU workload to share the same GPU's computing power and memory with Cannikin. To tune each node's computing power, we manually adjust the local dummy GPU workload's batch size to change the computing power and memory of Cannikin workloads. 

For the fastest node in cluster $C$, we allocate the entire RTX6000 GPU to the Cannikin container, while for the slowest node in cluster $C$, we restrict the performance of the node's Cannikin container to approximately one-fourth of an RTX6000 GPU. For the other intermediate nodes, we assign the batch size of each dummy workload to make them evenly distributed between the batch sizes of the fastest and slowest nodes. For example, in Cluster $C$ the batch size of the dummy workload in the slowest node is $150$, and in the fastest node is $0$. Then for the intermediate nodes, the dummy batch sizes are $10, 20,\dots, 140$. The results indicate that Cannikin's performance in Cluster $C$ aligns with that of Cluster $A$ and $B$. This brief experiment demonstrates the potential of Cannikin in addressing heterogeneity induced by resource sharing.

\smallskip
\noindent \textbf{Memory limitation.}
In the adaptive batch size training, the system will be initialized with a small total batch size~\cite{pollux} because when training starts, the relatively small batch size could guarantee statistical efficiency with less demand for the hardware resources. In Cannikin, this is a benefit to our system. In the experiments, we observe that powerful GPUs often outfit large GPU memory. On the other hand, weak GPUs' memory is usually relatively small. Since Cannikin evenly initializes each node's local mini batch size, if the initial total batch size is large, the nodes with weak GPU often reach their memory limitation, which causes training failure. Then in the following epoch, the total batch size increases. Cannikin will assign larger local mini batch sizes to the fast nodes and smaller local mini batch sizes to the slow nodes, which could avoid the GPU memory limit problem. We have also established local batch size constraints for each node to prevent out-of-memory failures.



\smallskip

\noindent \textbf{Adapt to schedulers for heterogeneous clusters.}
Existing dynamic resource allocating schedulers~\cite{pollux, optimus, slaq} only support the scheduling of homogeneous clusters. Sia~\cite{sia} is a scheduler with heterogeneity awareness, however, at each job level, the resource allocated for each job is still homogeneous. With Cannikin, we fully utilize the computing resources in heterogeneous clusters. When we design schedulers in future work, the scheduler should be able to allocate a heterogeneous cluster for each job, which can significantly increase resource utilization. 

With the performance metrics of Cannikin, the scheduler optimizes multi-job performance and reallocates resources for each job between epochs. In the multiple jobs scenario, if the scheduler removes nodes for Job $J$, Cannikin can easily use the learned computing models of remaining nodes along with the communication time model proposed by \citet{pollux} to predict and configure \textit{OptPerf} for Job $J$. When adding new nodes to Cluster $A$, Cannikin will re-initialize the cluster for job $J$ with two epochs. Then in the following epochs, Job $J$ will achieve \textit{OptPerf}.

\section{Other Related Work}

\textbf{Performance modeling of DNN training.} The importance of performance modeling in deep neural network training has been shown in recent research work. Paleo~\cite{Paleo} proposed the DNN training model by studying the NN topology. Clockwork~\cite{clockwork} model the GPU runtime with tracing. The All-reduce communication between nodes is studied and modeled~\cite{allreduceperf}. For the cloud-based DNN training~\cite{MLSYS2022_f457c545}, the accurate performance modeling and prediction significantly increase the training efficiency and reduce the clients' costs. From the scheduler perspective, an accurate DNN performance modeling~\cite{pollux, optimus, allox, gavel} increases the resource utilization and improves the fairness of multiple jobs execution. However, before Cannikin, none of the previous work synthetically considered the computing and communication model with the overlap in heterogeneous environments.

\smallskip

\noindent \textbf{Systems for adapted batch size training.} The philosophy of adaptive batch size training is choosing the batch size with the most convergence efficiency during training. Previous work~\cite{adjustbatch,adabatch} gradually increase the batch size and learning rate during training. Anytime Minibatch (AMB)~\cite{anytime} automatically tunes batch size by setting a fixed computing time for each batch, and each node computes as many samples as possible. However, AMB has only been evaluated on small applications and cannot support DNN workloads so far. The state-of-art adaptive batch size training system Pollux~\cite{pollux} can automatically configure the batch size by optimizing goodput during training. However, Pollux is designed for homogeneous environments. Cannikin is the first adapted batch size training system for heterogeneous environments.

\smallskip

\noindent \textbf{Accelerating ML on heterogeneous environments.} Most previous work about ML acceleration for heterogeneous clusters is on the scheduler level for multiple jobs. Hare~\cite{hare} optimized inter- and intra-job parallelism to increase the utilization of heterogeneous clusters. EasyScale~\cite{easyscale} dynamically assigned workers to scale distributed training for heterogeneous GPUs. Gandiva~\cite{gandiva} designed a strategy for heterogeneous GPU allocation that guarantees fairness and improves utilization. GPUlet~\cite{gpulet} designed an efficient scheduler by adopting spatial partitioning of GPU resources. However, looking into the single job level for the schedulers, the training strategy is still homogeneous. For job-level optimization on heterogeneous clusters, SnuHPL~\cite{snuhpl} improved the training in a heterogeneous HPC system by optimizing the data distribution for a given cluster configuration. BytePS~\cite{byteps} accelerated DNN training by leveraging CPU and bandwidth resources, however, BytePS focuses on the heterogeneity between CPU/GPU. Cannikin is a job-level-optimized system designed for heterogeneous GPU clusters. It automatically explores and determines the optimal local batch sizes assigned to each GPU and total batch sizes.

\section{Conclusion}
By modeling the GPU computing and communication jointly, we derive the optimal performance \textit{OptPerf} for heterogeneous clusters in the general case. Then we design Cannikin, which automatically predicts and configures \textit{OptPerf} for the heterogeneous cluster with high efficiency. Cannikin is the first adaptive distributed training system for heterogeneous clusters with near-optimal performance and high scalability by overcoming challenges such as optimal scenario determination and metrics measurement caused by heterogeneity. Cannikin outperforms state-of-art training systems for diverse workloads in real heterogeneous clusters. We believe Cannikin will inspire future DL training systems and DNN job schedulers in heterogeneous environments.

\bibliographystyle{ACM-Reference-Format}
\bibliography{sample-base}

\clearpage
\appendix
\section{Proof of Optimality Conditions}

\subsection{Compute-Bottleneck Scenario}
\label{app:comp_bottleneck}
\begin{proof}
When computing is the bottleneck for all nodes, the total processing time of one batch for the cluster is $\max_{i\in\mathcal{N}}\{t_{compute}^i + T_u\}$. Since $T_u$ is the same across all nodes, we can just consider $\max_{i\in\mathcal{N}}\{a_i+k_ib_i+m_i\}$. Minimizing this quantity is equivalent to solving the following optimization problem:
\begin{align*}
    \min\quad &\mu \\
    \text{s.t.}\quad &a_i+k_ib_i+m_i-\mu\leq 0, \quad \forall i \in \mathcal{N}\\
    &B - \sum_{i\in\mathcal{N}} b_i = 0.
\end{align*}
This optimization problem has corresponding Lagrangian:
\begin{equation*}
    L(\mu,\mathbf{b},\mathbf{\lambda},\nu) = \mu + \sum_{i\in\mathcal{N}}\lambda_i(a_i+k_ib_i+m_i-\mu)+\nu\left(B - \sum_{i\in\mathcal{N}} b_i\right)
\end{equation*}
Using the complimentary slackness conditions, we can solve and get $\lambda_i=(1/k_i)(\sum_{i\in\mathcal{N}}1/k_i)^{-1}$. The Karush-Kuhn-Tucker (KKT) conditions state that the optimal solution $\mu^*$ to this problem must satisfy $\lambda_i(a_i+k_ib_i+m_i - \mu^*) = 0$. Since $\lambda_i$ is strictly positive, $a_i+k_ib_i+m_i = \mu^*, \forall i \in \mathcal{N}$ so $t_{compute}^i = t_{compute}^j, \forall i,j \in \mathcal{N}$.
\end{proof}

\subsection{Communication-Bottleneck Scenario}
\label{app:comm_bottleneck}
\begin{proof}
When communication is the bottleneck for all nodes, the total processing time of one batch is $\max_{i\in\mathcal{N}}\{syncStart_i + T_{comm}\}$. Since $T_{comm}$ is the same across all nodes, we can just consider $\max_{i\in\mathcal{N}}\{a_i+\gamma(k_ib_i+m_i)\}$. Using the same technique as in the previous proof, we now get that for the optimal solution $\mu^*$, $a_i+\gamma(k_ib_i+m_i) = \mu^*, \forall i \in \mathcal{N}$. Thus \textit{OptPerf} must have $syncStart_i = syncStart_j, \forall i,j\in \mathcal{N}$.
\end{proof}


\subsection{General Optimal Scenario} 
\label{app:general_opt}
\begin{proof}
Let $\mathcal{N}_1$ be the set of computation-bottleneck nodes in $\mathcal{N}$ and let $\mathcal{N}_2$ be the set of communication-bottleneck nodes in $\mathcal{N}$. Minimizing the total cluster processing time of one batch is equivalent to the following optimization problem:
\begin{align*}
    \min\quad &\mu \\
    \text{s.t.}\quad &a_i+k_ib_i+m_i-\mu\leq 0, \quad \forall i \in \mathcal{N}_1\\
    &a_i+\gamma(k_ib_i+m_i)+T_o-\mu\leq 0, \quad \forall i \in \mathcal{N}_2\\
    &B - \sum_{i\in\mathcal{N}} b_i = 0.
\end{align*}
If we construct the Lagrangian for this problem, we see that by solving the complimentary slackness equations, the coefficients $\lambda_i$ are strictly positive for all $i$. Thus the optimal solution $\mu^*$ satisfies $a_i+k_ib_i+m_i = \mu^*$ for $i\in\mathcal{N}_1$ and satisfies $a_i+\gamma(k_ib_i+m_i)+T_o = \mu^*$ for $i\in\mathcal{N}_2$, giving the desired result.
\end{proof}

\section{The GNS in Heterogeneous Clusters}
\label{app:gns}

\textbf{Theorem~\ref{thm:opt_g_and_s}}. \emph{$\mathcal{G} = \sum_{i\in\mathcal{N}}w^{\mathcal{G}}_i \mathcal{G}_i$ and $\mathcal{S} = \sum_{i\in\mathcal{N}}w^{\mathcal{S}}_i \mathcal{S}_i$ are minimum variance, unbiased linear estimators of $|G|^2$ and $\text{tr}(\Sigma)$ when:}
\begin{align*}
        \mathbf{w}^{\mathcal{G}} = \frac{\mathbf{1}^T A_{\mathcal{G}}^{-1}}{\mathbf{1}^T A_{\mathcal{G}}^{-1}\mathbf{1}},\quad \quad \mathbf{w}^{\mathcal{S}} = \frac{\mathbf{1}^T A_{\mathcal{S}}^{-1}}{\mathbf{1}^T A_{\mathcal{S}}^{-1}\mathbf{1}},\label{eq:opt_g_and_s}
    \end{align*}
\emph{where $\mathbf{1}$ is an $n$-dimensional column vector of ones and both $A_{\mathcal{G}}$ and $A_{\mathcal{S}}$ are $n\times n$ matrices with respective entries $a_{\mathcal{G}}(i,j)$ and $a_{\mathcal{S}}(i,j)$:}
    \begin{equation*}
        a_{\mathcal{G}}(i,i) = \dfrac{B+2b_i}{B^2-Bb_i}, \quad a_{\mathcal{G}}(i,j) = \dfrac{B^2-b_i^2-b_j^2}{B(B-b_i)(B-b_j)} \text{ for } i\neq j
    \end{equation*}
    \begin{equation*}
        a_{\mathcal{S}}(i,i) = \dfrac{Bb_i}{B-b_i}, \quad a_{\mathcal{S}}(i,j) = \dfrac{b_i b_j(B-b_i-b_j)}{(B-b_i)(B-b_j)} \text{ for } i\neq j
    \end{equation*}

\begin{proof}
Since $\mathcal{G}_i$ is an unbiased estimator of $|G|^2$, $\mathcal{G}$ is an unbiased estimator when $\sum_{i\in\mathcal{N}}w_i^{\mathcal{G}} = 1$. Furthermore $\mathcal{G}$ is the minimum variance, unbiased linear estimator when $\mathbf{w}$ minimizes the quadratic form of $\Sigma(\mathcal{G}_i)$, where $\Sigma(\mathcal{G}_i)$ is the correlation matrix of the estimators $\mathcal{G}_i$. Using Lagrange multipliers, we get:
\begin{equation*}
    \mathbf{w}^{\mathcal{G}} = \frac{\mathbf{1}^T \Sigma(\mathcal{G}_i)^{-1}}{\mathbf{1}^T \Sigma(\mathcal{G}_i)^{-1}\mathbf{1}}
\end{equation*}
Similarly, we get that $\mathcal{S}$ is the minimum variance, unbiased linear estimator of $\text{tr}(\Sigma)$ when:
\begin{equation*}
    \mathbf{w}^{\mathcal{S}} = \frac{\mathbf{1}^T \Sigma(\mathcal{S}_i)^{-1}}{\mathbf{1}^T \Sigma(\mathcal{S}_i)^{-1}\mathbf{1}}
\end{equation*}
where $\Sigma(\mathcal{S}_i)$ is the covariance matrix of the estimators $\mathcal{S}_i$.

To compute $\mathbf{w}^{\mathcal{G}}$, we require $\Sigma(\mathcal{G}_i)$. By definition, the matrix's diagonal elements are $\text{Var}(\mathcal{G}_i)$ and the off-diagonal elements are $\text{Cov}(\mathcal{G}_i,\mathcal{G}_j)$ for $i\neq j$. Lemma~\ref{lemma:gi_and_si_var} gives us $\text{Var}(\mathcal{G}_i)$ and Lemma~\ref{lemma:cov_gi_gj} gives us $\text{Cov}(\mathcal{G}_i,\mathcal{G}_j)$. Since all terms of $\Sigma(\mathcal{G}_i)$ have a common factor of $4|G|^2\text{tr}(\Sigma)$, this factor will cancel for $\mathbf{w}^{\mathcal{G}}$, so we can equivalently solve for $\mathbf{w}^{\mathcal{G}}$ using the matrix $A_{\mathcal{G}}$ instead of $\Sigma(\mathcal{G}_i)$, where:
    \begin{equation*}
        a_{\mathcal{G}}(i,i) = \dfrac{B+2b_i}{B^2-Bb_i}, \quad a_{\mathcal{G}}(i,j) = \dfrac{B^2-b_i^2-b_j^2}{B(B-b_i)(B-b_j)} \text{ for } i\neq j
    \end{equation*}
We can use a similar argument for $\mathcal{S}$ with Lemmas~\ref{lemma:var_g} and~\ref{lemma:cov_si_sj}, where rather than using $\Sigma(\mathcal{S}_i)$ we can use the matrix $A_{\mathcal{G}}$ with entries:
    \begin{equation*}
        a_{\mathcal{S}}(i,i) = \dfrac{Bb_i}{B-b_i}, \quad a_{\mathcal{S}}(i,j) = \dfrac{b_i b_j(B-b_i-b_j)}{(B-b_i)(B-b_j)} \text{ for } i\neq j
    \end{equation*}

\end{proof}

\begin{lemma}
    \label{lemma:gi_and_si_var}
    The estimators $\mathcal{G}_i$ and $\mathcal{S}_i$ have variances:
    \begin{align*}
        \text{Var}(\mathcal{G}_i) &= 4|G|^2\text{tr}(\Sigma)\left(\frac{B+2b_i}{B^2-Bb_i}\right)\\[0.5ex]
        \text{Var}(\mathcal{S}_i) &= 4|G|^2\text{tr}(\Sigma)\left(\frac{Bb_i}{B-b_i}\right)
    \end{align*}
    where $\text{tr}(\Sigma)$ is the sum of the variances of the individual gradient components and $|G|^2$ is the global norm of the gradient.
\end{lemma}

\begin{proof}
First we compute the variance of $\mathcal{G}_i$:
\begin{align*}
    \text{Var}(\mathcal{G}_i) = \text{Var}(\frac{B}{B-b_i}|g|^2 - \frac{b_i}{B-b_i}|g_i|^2) =\\
    \stackrel{(1)}{=} \left(\frac{B}{B-b_i}\right)^2 \text{Var}(|g|^2) + \left(\frac{b_i}{B-b_i}\right)^2 \text{Var}(|g_i|^2) - \\
    - 2\left(\frac{B}{B-b_i}\right)\left(\frac{b_i}{B-b_i}\right)\text{Cov}(|g|^2,|g_i|^2) = \\
    \stackrel{(2)}{=} \left(\frac{B}{B-b_i}\right)^2 \cdot \frac{4|G|^2\text{tr}(\Sigma)}{B} + \left(\frac{b_i}{B-b_i}\right)^2 \cdot \frac{4|G|^2\text{tr}(\Sigma)}{b_i} - \\ 
    - \frac{2Bb_i}{(B-b_i)^2} \text{Cov}(|g|^2,|g_i|^2) = \\
    \stackrel{(3)}{=} \frac{B+b_i}{(B-b_i)^2}\left(4|G|^2\text{tr}(\Sigma)\right) - \frac{2Bb_i}{(B-b_i)^2} \cdot \frac{4b_i|G|^2\text{tr}(\Sigma)}{B^2} =\\
    = 4|G|^2\text{tr}(\Sigma)\left(\frac{B+b_i}{(B-b_i)^2}-\frac{2b_i^2}{B(B-b_i)^2}\right) =\\
    = 4|G|^2\text{tr}(\Sigma)\left(\frac{B+2b_i}{B^2-Bb_i}\right)
\end{align*}
where $(1)$ follows from the variance of sums of random variables, $(2)$ follows from Lemma~\ref{lemma:var_g} and $(3)$ follows from Lemma~\ref{lemma:cov_g_g_i}. 
We can similarly compute the variance of $\mathcal{S}_i$:
\begin{align*}
    \text{Var}(\mathcal{S}_i) = \text{Var}(\frac{b_iB}{B-b_i}|g_i|^2 - \frac{b_iB}{B-b_i}|g|^2) =\\
    \stackrel{(1)}{=} \left(\frac{b_i B}{B-b_i}\right)^2 \left(\text{Var}(|g_i|^2) + \text{Var}(|g|^2) - 2\text{Cov}(|g|^2,|g_i|^2)\right) = \\
    \stackrel{(2)}{=} \left(\frac{b_i B}{B-b_i}\right)^2 \left(\frac{4|G|^2\text{tr}(\Sigma)}{b_i} + \frac{4|G|^2\text{tr}(\Sigma)}{B} - 2\text{Cov}(|g|^2,|g_i|^2) \right) = \\ 
    \stackrel{(3)}{=} 4|G|^2\text{tr}(\Sigma)\left(\frac{b_i B}{B-b_i}\right)^2\left(\frac{1}{b_i}+\frac{1}{B}-\frac{2b_i}{B}\right) = \\
    = 4|G|^2\text{tr}(\Sigma)\left(\frac{Bb_i}{B-b_i}\right)
\end{align*}
where again $(1)$ follows from the variance of sums of random variables, $(2)$ follows from Lemma~\ref{lemma:var_g} and $(3)$ follows from Lemma~\ref{lemma:cov_g_g_i}.
\end{proof}

\begin{lemma}
    \label{lemma:cov_gi_gj}
    The estimators $\mathcal{G}_i$ and $\mathcal{G}_j$ have covariance:
    \begin{equation*}
        \text{Cov}(\mathcal{G}_i,\mathcal{G}_j) = 4|G|^2\text{tr}(\Sigma) \frac{B^2-b_i^2-b_j^2}{B(B-b_i)(B-b_i)}
    \end{equation*}
\end{lemma}

\begin{proof}
\small
    Using the definition of $\mathcal{G}_i$ and $\mathcal{G}_j$:
    \begin{align*}
        \text{Cov}(\mathcal{G}_i,\mathcal{G}_j) = \text{Cov}\left(\frac{B|g|^2-b_i|g_i|^2}{B-b_i},\frac{B|g|^2-b_j|g_j|^2}{B-b_j}\right) =\\
        \stackrel{(1)}{=} \frac{B^2\text{Var}(|g|^2)}{(B-b_i)(B-b_j)} - \frac{Bb_i\text{Cov}(|g|^2,|g_i|^2)}{(B-b_i)(B-b_j)} - \frac{Bb_j\text{Cov}(|g|^2,|g_j|^2)}{(B-b_i)(B-b_j)} =\\
        \stackrel{(2)}{=} \frac{4|G|^2\text{tr}(\Sigma)}{(B-b_i)(B-b_j)}\left(\frac{1}{B} - \frac{b_i^2}{B}-\frac{b_j^2}{B}\right) 
    \end{align*}
    where (1) follows from the covariance of linear combinations of random variables and the independence of $g_i$ and $g_j$, and (2) follows from Lemma~\ref{lemma:var_g} and Lemma~\ref{lemma:cov_g_g_i}.
\end{proof}

\begin{lemma}
    \label{lemma:cov_si_sj}
    The estimators $\mathcal{S}_i$ and $\mathcal{S}_j$ have covariance:
    \begin{equation*}
        \text{Cov}(\mathcal{S}_i,\mathcal{S}_j) = 4|G|^2\text{tr}(\Sigma)\frac{b_ib_j(B-b_i-b_j)}{(B-b_i)(B-b_j)}
    \end{equation*}
\end{lemma}

\begin{proof}
\small
    Using the definition of $\mathcal{S}_i$ and $\mathcal{S}_j$:
    \begin{align*}
        \text{Cov}(\mathcal{S}_i,\mathcal{S}_j) = \text{Cov}\left(\frac{Bb_i}{B-b_i}(|g_i|^2-|g|^2),\frac{Bb_j}{B-b_j}(|g_j|^2-|g|^2) \right) =\\
        \stackrel{(1)}{=} \frac{B^2b_ib_j}{(B-b_i)(B-b_j)}\left(\text{Var}(|g|^2) - \text{Cov}(|g|^2,|g_i|^2) - \text{Cov}(|g|^2,|g_i|^2)\right) =\\
        \stackrel{(2)}{=} \frac{4|G|^2\text{tr}(\Sigma)B^2b_ib_j}{(B-b_i)(B-b_j)}\left(\frac{1}{B} - \frac{b_i}{B^2} - \frac{b_j}{B^2}\right)
    \end{align*}
    where (1) follows from the covariance of linear combinations of random variables and the independence of $g_i$ and $g_j$, and (2) follows from Lemma~\ref{lemma:var_g} and Lemma~\ref{lemma:cov_g_g_i}.
\end{proof}

\begin{lemma}
    \label{lemma:var_g}
    For any estimated gradient $g_{est}$ with corresponding batch size $b$, the variance of the gradient norm satisfies:
    \begin{equation*}
        \text{Var}(|g_{est}|^2) \approx \frac{4|G|^2\text{tr}(\Sigma)}{b}
    \end{equation*}
\end{lemma}

\begin{proof}
    Error propagation using Taylor's rule~\cite{oehlert1992note} (also known as the delta method) gives us the approximation:
    \begin{equation*}
        \text{Var}(|g_{est}|^2) \approx 4\mathbb{E}[|g_{est}|]^2 \text{Var}(|g_{est}|) = 4|G|^2\cdot \frac{1}{b}\text{tr}(\Sigma)
    \end{equation*}
\end{proof}

\begin{lemma}
    \label{lemma:cov_g_g_i}
    For local gradient $g_i$ at node $i$ with batch size $b_i$ and global gradient $g$ with batch size $B$ and computed using Eq~\eqref{eq:grad_agg}, the covariance of the two gradient norms is:
    \begin{equation*}
        \text{Cov}(|g|^2,|g_i|^2) = \frac{4b_i|G|^2\text{tr}(\Sigma)}{B^2}
    \end{equation*}
\end{lemma}

\begin{proof}
    The global gradient norm $|g|^2$ can be written in terms of $g_i$ and non-$g_i$ components:
    \begin{align*}
        |g|^2 &= \frac{1}{B^2}\sum_{j\in b_i} (\nabla_{\theta}L_{x_j}(\theta))^2 + \frac{1}{B^2}\sum_{j\notin b_i} (\nabla_{\theta}L_{x_j}(\theta))^2 =\\
        &= \frac{b_i^2}{B^2}|g_i|^2 + \frac{1}{B^2}\sum_{j\notin b_i} (\nabla_{\theta}L_{x_j}(\theta))^2
    \end{align*}
    Rewriting the covariance using this expression:
    \begin{align*}
        \text{cov}(|g|^2,|g_i|^2) = \text{cov}(\frac{b_i^2}{B^2}|g_i|^2 + \frac{1}{B^2}\sum_{j\notin b_i} (\nabla_{\theta}L_{x_j}(\theta))^2,|g_i|^2) =\\
        \stackrel{(1)}{=} \frac{b_i^2}{B^2}\text{cov}(|g_i|^2,|g_i|^2) + \frac{1}{B}\text{cov}(\sum_{j\notin b_i} (\nabla_{\theta}L_{x_j}(\theta))^2,|g_i|^2) = \\
        \stackrel{(2)}{=} \frac{b_i^2}{B^2}\text{Var}(|g_i|^2) \stackrel{(3)}{=} \frac{b_i^2}{B^2} \cdot \frac{4|G|^2\text{tr}(\Sigma)}{b_i} = \frac{4b_i|G|^2\text{tr}(\Sigma)}{B^2}
    \end{align*}
    where $(1)$ follows from the covariance of linear combinations of random variables, $(2)$ follows from the variance-covariance relationship and the independence of $g_i$ and $g_j$ for $i\neq j$ and $(3)$ follows from Lemma~\ref{lemma:var_g}.
\end{proof}









\end{document}